\documentclass[11pt]{article}  
\usepackage{amsmath,amssymb,amsthm,amsfonts,verbatim}
\usepackage{fullpage,geometry,xcolor}

\title{Pseudorandomness for concentration bounds and signed majorities}
\author{Parikshit Gopalan \\Microsoft \and Daniel M. Kane \\University of California, San Diego \and Raghu Meka}

\newcommand{\floor}[1]{\left\lfloor #1 \right\rfloor}

\newcommand{\supp}{\textrm{supp}}

\newcommand{\eat}[1]{}
\newcommand{\calG}{\ensuremath{\mathcal{G}}}

\newcommand{\E}{\mathop{\mathbb{E}\/}}

\newcommand{\poly}{\ensuremath{\mathrm{poly}}}
\newcommand{\polylog}{\ensuremath{\mathrm{polylog}}}

\newcommand{\prg}{{\ensuremath \mathsf{PRG}}}
\newcommand{\bpp}{{\ensuremath \mathsf{BPP}}}
\newcommand{\rl}{{\ensuremath \mathsf{RL}}}
\newcommand{\logspace}{{\ensuremath \mathsf{L}}}
\newcommand{\ptime}{{\ensuremath \mathsf{P}}}
\newcommand{\tc}{{\ensuremath \mathsf{TC}_0}}

\newcommand{\eps}{\ensuremath{\varepsilon}}
\newcommand{\zo}{\{0,1\}}
\newcommand{\pmo}{\ensuremath \{ \pm 1\}}
\newcommand{\rgta}{\ensuremath{\rightarrow}}

\newcommand{\R}{\mathbb R}

\newcommand{\pr}{\Pr}
\newcommand{\sgn}{\textrm{sgn}}
\newcommand{\dotp}[2]{\left\langle #1,#2\right\rangle}

\newtheorem{thm}{Theorem}
\newtheorem{definition}{Definition}
\newtheorem{prop}[thm]{Proposition}

\newtheorem{lem}[thm]{Lemma}

\newtheorem{claim}{Claim}
\newtheorem*{defn}{Definition}

\newtheorem{fact}[thm]{Fact}

\newcommand{\zpm}{\{-1,0,1\}}
\newcommand{\dpm}{\pmo}
\newcommand{\gcs}{G^{CS}}
\newcommand{\glarge}{G^{b}}

\newcommand{\dtv}{d_{TV}}
\newcommand{\hh}{\mathcal{H}}
\newcommand{\ignore}[1]{{}}
\newcommand{\bkets}[1]{\left(#1\right)}
\newcommand{\sbkets}[1]{\left[#1\right]}

\newcommand{\iprod}[2]{\langle #1,\ #2\rangle}
\newcommand{\nmo}[1]{\left\|#1\right\|_1}
\newcommand{\nmt}[1]{\left\|#1\right\|_2}
\newcommand{\nm}[1]{\left\|#1\right\|}
\newcommand{\nmp}[1]{\left\|#1\right\|_p}

\newcommand{\abs}[1]{\left|#1\right|}
\newcommand{\calD}{\mathcal{D}}

\newcommand{\rn}{\mathbb{R}^n}
\newcommand{\mnote}[1]{}
\newcommand{\pnote}[1]{}
\newcommand{\dnote}[1]{}

\newcommand{\rom}{signed majorities }
\newcommand{\Rom}{Signed majorities }
\date{}
\begin{document}
\begin{titlepage}
\thispagestyle{empty}
\maketitle

\begin{abstract}
The problem of constructing pseudorandom generators that fool
halfspaces has been studied intensively in recent times. For fooling halfspaces over $\pmo^n$ with polynomially small
error, the best construction known requires seed-length $O(\log^2(n))$
\cite{MekaZ13}. Getting the seed-length down to $O(\log(n))$ is a natural
challenge in its own right, which needs to be overcome in order to
derandomize $\rl$. In this work we make progress towards this goal by obtaining
near-optimal generators for two important special cases:
\begin{itemize}  

\item We give a near optimal derandomization of the Chernoff bound for
  independent, uniformly random bits. Specifically, we show how to
  generate $x \in \pmo^n$ using $\tilde{O}(\log (n/\eps))$ random bits
  such that for any unit vector $u$, $u\cdot x$ matches the sub-Gaussian tail
  behaviour predicted by the Chernoff bound up to error $\eps$.

\item We construct a generator which fools halfspaces with $\{0,1,-1\}$
  coefficients with error $\epsilon$ with a seed-length of
  $\tilde{O}(\log(n/\epsilon))$. This includes the important special
  case of majorities.
\end{itemize}

In both cases, the best previous results required seed-length of $O(\log
n + \log^2(1/\epsilon))$.    

Technically, our work combines new Fourier-analytic tools with the
iterative dimension reduction techniques and the gradually increasing
independence paradigm  of previous works \cite{KaneMN11, CelisRSW13,
  GopalanMRTV12}. 
\end{abstract}
\end{titlepage}
\section{Introduction}
The theory of pseudorandomness has given compelling evidence that very strong
pseudorandom generators ($\prg$s) exist.  For example, assuming that
there are computational problems solvable in exponential time that
require exponential-sized circuits, Impagliazzo and
Wigderson~\cite{ImpagliazzoWi97} showed that there exist very strong $\prg$s which
allow us to simulate every randomized algorithm deterministically with
only a polynomial slowdown, and thus $\bpp= \ptime$. 
These results, however, are conditional on a circuit complexity
assumption whose proof seems far off.  Since $\prg$s that
fool a class of Boolean circuits also imply lower bounds for that class, we
cannot hope to circumvent this assumption. Thus unconditional generators
are only possible for  restricted models of computation for which we have strong lower bounds.

{\em Bounded-space algorithms} are a natural computational model
for which we know how to construct strong $\prg$s unconditionally.
Let $\rl$ denote the class of randomized algorithms with $O(\log n)$
work space which can access the random bits in a read-once pre-specified order.
Nisan~\cite{Nisan92} devised a $\prg$ of seed
length $O(\log^2(n/\eps))$ that fools $\rl$.
This generator was used by Nisan himself to show that $\rl
\subseteq \mathsf{SC}$ \cite{Nisan94} and by Saks and Zhou~\cite{SaksZ99} to
prove that $\rl$ can be simulated in space $O(\log^{3/2} n)$. 
Constructing $\prg$s with the optimal
$O(\log(n/\epsilon))$ seed length for this class and showing that $\rl
= \logspace$ is arguably the outstanding open problem in
derandomization (which might not require a breakthrough in lower bounds). 
Despite much progress in this area
\cite{ImpagliazzoNW94,NisanZ96,RazR99,Reingold08,ReingoldTV06,BravermanRRY14,BrodyV10,KouckyNP11,De11,GopalanMRTV12},
there are few cases where we can improve on Nisan's twenty year old
bound of $O(\log^2 n)$ \cite{Nisan92}.

Halfspaces are Boolean functions $h:\pmo^n \to \dpm$ described as $h(x)
= \sgn(\dotp{w}{x} - \theta)$ for some \emph{weight} vector $w \in
\R^n$ and \emph{threshold} $\theta \in \R$. They are of central importance
in computatonal complexity, learning theory and social choice. Lower
bounds for halfspaces are trivial, whereas the problem of proving
lower bounds against depth-$2$ $\tc$ or halfspaces of
halfspaces is a frontier open problem in computational complexity. The 
problem of constructing explicit $\prg$s that can fool halfspaces is a
natural challenge that has seen a lot of exciting progress recently
\cite{DGJSV09, MekaZ13, Kane11,Kane14}. The best known $\prg$
construction for halfspaces is that of Meka and Zuckerman
\cite{MekaZ13} who gave a $\prg$ with seed-length $O(\log n +
\log^2(1/\epsilon))$, which is $O(\log^2(n))$ for polynomially small
error. They also made a connection to space bounded algorithms by 
showing that $\prg$s against $\rl$ with
inverse polynomial error can be used to fool halfspaces. Thus
constructing better $\prg$s for halfspaces seems to be a
necessary step towards progress for bounded-space algorithms. 

Beyond computational complexity, the problem of
constructing better $\prg$s for halfspaces has ample algorithmic
motivation;  perhaps the most compelling of which comes from the
ubiquitous applications in computer science of Chernoff-like bounds
for weighted sums of the form $\sum_iw_ix_i$ where the $x_i$s are
uniformly random bits. There has been a long line of work on showing sharp tail bounds for pseudorandom sequences starting from
\cite{SchmidtSS95}. A $\prg$ for halfspaces with seed-length $O(\log(n/\eps))$
would give a space of support size $\poly(n)$ where Chernoff-like tail
bounds hold. This in turn would yield a black-box derandomization with
only a polynomial slow-down of any algorithm which relies on uniform
randomness only for such tail bounds. $\prg$s for halfspaces also have
other algorithmic applications to streaming algorithms for duplicate
detection \cite{GopalanR09} and efficient revenue maximization for
certain kinds of auctions \cite{GopalanNR14}.

\subsection{Our results}

A $\prg$ is a function $\calG:\pmo^r \rgta \pmo^n$. We refer to $r$ as
the seed-length of the generator. The $\tilde{O}()$ notation hides
polylogarithmic factors in its argument. We say $\calG$ is \emph{explicit} if
the output of $\calG$ can be computed in time $\poly(n)$.

\begin{definition}
A $\prg$ $\calG:\pmo^r \to \dpm^n$ fools a class of
functions $\mathcal{F} = \{f:\dpm^n \to \dpm\}$ with error $\eps$
(or $\eps$-fools $\mathcal{F}$) if for every $f \in \mathcal{F}$, 
\[\left|\pr_{x \in \dpm^n}[f(x) = 1] - \pr_{y \in \pmo^r}[f(G(y)) = 1]\right| < \eps.\]
\end{definition}

\subsection*{Derandomized Chernoff bounds}
Chernoff bounds are a basic tool in the analysis of randomized
algorithms. A ubiquitous version that applies to the setting of
independent random bits is the following:

\begin{claim}[Chernoff bound]
There exist constants $c_1, c_2 > 0$ such that for every unit vector $w \in
\R^n$ and $t \geq 1$, $\Pr_{x \in \pmo^n}\left[\, \left|\iprod{w}{x}\right| > t \right] \leq
 c_1e^{-c_2 t^2}$.
\end{claim}

We obtain a near-optimal derandomization of this result. 

\begin{thm}
\label{thm:chernoff}
There exists an explicit generator $\calG_1:\pmo^r \to \dpm^n$
and constants $d_1, d_2$ such that for every unit vector $w \in \R^n$,
$t \geq 1$ and $\eps > 0$,
\[ 
\Pr_{y \in_u \pmo^r}\left[\,\left|\dotp{w}{\calG_1(y)}\right| > t \right] \leq d_1e^{-d_2t^2} + \eps.
\]
The generator has seed-length $r =
\tilde{O}(\log(n/\eps))$.
\end{thm}

To contrast this with what was known previously, consider the setting
where $\eps =1/\poly(n)$. The Chernoff bound asserts that the
probability that $|\iprod{w}{x}| = \Omega(\sqrt{\log(n)})$ is inverse
polynomially small. A $\prg$ for halfspaces with error parameter $\eps
=1/\poly(n)$ would also gurarantee such tails, but the best known
construction requires seed-length $O(\log^2(n))$ \cite{MekaZ13}. One could also get
such tail bounds using limited indpendence \cite{SchmidtSS95}; however, we
would need $O(\log(n))$-wise independence, which again requires
$O(\log^2(n))$ seed-length.

\subsubsection*{Fooling signed majorities}

An important sub-class of halfspaces are those whose weight
vectors  have $\{0,1,-1\}$-valued entries. This corresponds to
selecting a subset of variables, assiging each of them an orientation and then
taking a threshold. We henceforth refer to this class of 
halfspaces as signed majorities. \Rom arise naturally in voting theory,
learning theory and property testing - see ~\cite{MosselOO05,MatulefORS09, RonS13, BlaisOD10}. Fooling such tests
requires fooling the sum of arbitrary subsets of variables in statistical distance,
a problem that was studied by \cite{GopalanMRZ13} in their work on
fooling {\em combinatorial shapes}. 
Fooling sums in statistical distance includes as a special case modular
tests on sums of variables with unrestricted modulus \cite{LovettRTV09,MekaZ09}. $\prg$s for
modular sums are a strong generalization of the versatile
small-bias spaces \cite{NaorN93} which correspond to fooling modular
sums with modulus two. The best previously known $\prg$s due to
Lovett {\em et al.} for such tests require seed-length $O(\log^2 n)$
\cite{LovettRTV09} for large modulii, but their result can also
handle sums with non-binary coefficients. Finally, \rom seem to capture several
technical hurdles in designing optimal $\prg$s
for halfspaces. 

We construct  a $\prg$ which $\eps$-fools \rom with a
seed-length of $\tilde{O}(\log (n/\eps))$.  

\begin{thm}\label{th:mainintrohs}
There exists an explicit generator $\calG_2:\pmo^r \to \dpm^n$ with seed-length $r = \tilde{O}(\log(n/\eps))$ which $\eps$-fools \rom.
\end{thm}

The best previous result even for \rom had a
seed-length of $O(\log n + \log^2(1/\eps))$ \cite{MekaZ13}. For
the important case of polynomially small error, $\eps =
1/\poly(n)$, our result gives the first improvement over the $O(\log^2
n)$ bound implied by directly applying known $\prg$s for space-bounded
machines  \cite{Nisan92, ImpagliazzoNW94}. 

Independently and concurrently De \cite{De14} gave a $\prg$ for
\emph{combinatorial shapes} introduced by \cite{GopalanMRZ13} with a seed-length of
$O(\log^{3/2}(n/\eps))$. These objects are more general than \rom but De's seed-length is
worse than ours.


\subsection{Other related work}\label{sec:previouswork}
Starting with the work of Diakonikolas et al.~\cite{DGJSV09}, there has been a lot of work on
constructing PRGs for halfspaces and related classes of intersections
of halfspaces and polynomial threshold functions over the domain $\pmo^n$
\cite{DiakonikolasKN10, GopalanOWZ10, HarshaKM12, MekaZ13, Kane11,
  Kane11b,Kane14}. Rabani and Shpilka \cite{RabaniS09} construct optimal
hitting set generators for halfspaces over $\pmo^n$; hitting set
generators are in general weaker than $\prg$s however.  

Another line of work gives constructions of
$\prg$s for halfspaces for the uniform distribution over
the sphere (\emph{spherical caps}) or the Gaussian distribution. This
case is  easier than constructing $\prg$s for halfspaces over
the hypercube; the latter objects are known to imply the former with comparable parameters. For spherical caps, Karnin, Rabani and Shpilka \cite{KarninRS12} gave
a $\prg$ with a seed-length of $O(\log n + \log^2(1/\eps))$. For the
Gaussian distribution, \cite{Kane14} gave a $\prg$ which achieves a
seed-length of $O(\log n + \log^{3/2}(1/\eps))$. Very recently,
Kothari and Meka \cite{KothariM14} gave a $\prg$ for spherical caps with
a seed-length of $\tilde{O}(\log(n/\eps))$. At a high level, \cite{KothariM14} also uses
the iterative dimension reduction approach like in \cite{KaneMN11,
  CelisRSW13, GopalanMRTV12}; however, the final construction and its
analysis are significantly different.  

\subsection{Overview of our constructions}\label{sec:overview}
\subsubsection*{Derandomized Chernoff bounds}

Our first attempt at constructing a $\prg$ for the Chernoff bound
applies a simple dimension reduction step iteratively.
\begin{enumerate}
\item Starting from a linear function $\sum_{i=1}^n w_ix_i$,
  (pseudo)randomly hash the variables into $\sqrt{n}$ buckets using a
  hash function $h$.
\item Use an $\eps$-biased string $x$ to sum up coefficients within a
  bucket. This gives a new linear function $\sum_{j=1}^{\sqrt{n}}v_jy_j$
  in $\sqrt{n}$ dimensions where $v_j = \sum_{i:h(i) =j}w_ix_i$.
\end{enumerate}

Repeating this step $\log\log(n)$ times, we get down to $\theta \in
\R$ which is the value we output. Call this generator $\calG'$. It is
easy to see that each output bit of $\calG'$  is the xor of $\log\log (n)$ bits from independent $\eps$-biased strings, where the hash
functions are used to select co-ordinates from each string. This
technique of applying pseudorandom dimension reduction iteratively is similar to
\cite{KaneMN11,CelisRSW13, GopalanMRTV12}.  

Does this generator give the desired tail behavior? Assume that we start from a
unit vector $w \in \R^n$.  To get tail bounds, we would like to
control the $\ell_2$ norm, which starts at $1$ but could increase
substantially for particular choices of $x$. The Chernoff bound says
that for truly random $x$, the $\ell_2$ norm is unlikely to increase
by more than a factor of $c\sqrt{\log(n)}$. Even if we manage to match
this tail behavior in each step by choosing $x$ pseudorandomly (which
is the problem we are trying to solve), the final bound we get would
be $O((\log n)^{\log\log(n)/2})$. Using $\eps$-biased $x$, we show a
weaker bound of $\polylog(n)$ for each step, giving an overall bound of
$d(n) = (\log(n))^{O(\log\log(n))}$. Showing this bound for one step requires a fair
amount of technical work, it works by decomposing the vector into weight scales
and tuning the amount of independence to the scale like in
\cite{GopalanMRTV12}. We leave open the question of whther $\calG'$
can itself give Chernoff-like tail bounds. 

Next we show that one gets the desired tail behaviour by hashing
variables into $m = \poly(d(n))$ buckets and using an independent copy of
$\calG'$  for each bucket. The reason is  the output of the resulting
generator can be viewed as the sum of $m$ independent {\em bounded}
random variables, which lets us apply Bernstein's inequality which guarantees
Chernoff-like tails for such variables. The boundedness comes from the
tail guarantee of $\calG'$: since large deviations are very unlikely,
we can condition on the event that they do not occur in any of the
buckets. The final step is to reduce to seed-length, we do
this by recycling the seed for the various independent copies of $\calG'$
using the INW generator \cite{ImpagliazzoNW94}, like in \cite{MekaZ13}. 

\subsubsection*{Fooling \rom}
Let us fix a test vector $v \in \zpm^n$ and error $\eps = 1/\poly(n)$.
Fooling \rom with polynomially small error is equivalent to fooling
\emph{linear sums} of the form $\dotp{v}{x}$ in statistical distance
with error $1/\poly(n)$. We shall adopt this view from now on.  


We start with a generator that uses iterated dimension reduction
and gradually-increasing independence as we did for derandomizing the
Chernoff bound. This by itself is not enough for
fooling sums in statistical distance. The reason is that there
exist small-bias spaces with exponentially small bias that are far from fooling linear sums in statistical distance, like the set of strings whose weight is divisible by $3$ \cite{ViolaW07}. 
We design a different generator to deal with such tests and then combine the two generators by xoring independent copies.

Next, note that showing closeness in statistical distance for
discrete random variables is equivalent to showing that their
\emph{Fourier transforms} are close. Using this, it suffices
to design a generator $\calG:\pmo^r \to \pmo^n$ such that for all
$\alpha \in \R$, the corresponding \emph{Fourier coefficient} $\E_y[\exp(2\pi i \alpha \dotp{v}{\calG(y)})]$ is
close to its value under the uniform distribution. Note that in order
to fool the mod $m$ test, it suffices to fool all $\alpha = j/m$ for
integers $j$.  We consider two cases based on how large $\alpha$ is
relative to $\|v\|_0 = k$.\\

\eat{The two cases we consider will capture a
shift in the quantitative behaviour of the Fourier coefficients and
also implicitly take care of the technical issues with using iterative
dimension reduction mentioned above. \pnote{This is too vague: can
  we be more concrete?}}

{\bf Large $\alpha$: } Here we consider $\alpha \gg 1/\sqrt{k}$. This includes the case
of modular tests where the modulus is much smaller than $\sqrt{k}$. We fool such tests using
an error reduction procedure. We start with the generator of \cite{GopalanMRZ13} which requires seed-length $O(\log n)$ to fool
such tests with constant error. We then reduce the error to inverse
polynomial at the expense of a $O(\log \log n)$ factor
in seed-length using standard machinery from pseudorandomness. 
While technically simple, this step is the botteleneck
in extending our result to more general halfpsaces: there is no analog
of the \cite{GopalanMRZ13} generator to start from.\\

{\bf Smaller $\alpha$: }This case which includes modular tests where the
modulus is $\Omega(\sqrt{k})$ is the harder case and technically the most novel
portion of this work. The qualitative difference from the other case
can be seen from the fact that when we sum $k$ random bits modulo $m = \omega(\sqrt{k})$, the
resulting distribution is no longer uniform over congruence classes.

The generator uses dimension reduction in a manner similar to what we used to derandomize the
Chernoff bound. Like before, the plan is to show that a single dimension reduction step does not incur too much error. However, the analysis is very different and requires several new tools. This step critically exploits the recursive structure of the generator: to analyze the error we can work as if the variables in the reduced space are given truly random signs and then recursively analyze the error in the reduced space. Working with truly random bits in the smaller-dimensional space helps us reduce bounding the error to finding good low-degree polynomial approximators for a certain product of cosines. In the most technically, involved part of our argument we use various analytic tools to find such low-degree approximators. One additional ingredient is that the above approach does not actually work for all test vectors but only for \emph{sufficiently well-spread out} vectors as measured by their $\ell_2$, $\ell_4$ norms. The final piece is to argue that the $\ell_2,\ell_4$ norms are not distorted too much by the dimension reduction steps.

\ignore{
Ideally we would have liked to make such a claim for all test
xfvectors, but this turns out to be false. What is true is that a single
dimension reduction step does not incur too much error when the test
vector is \emph{sufficiently well-spread out} as measured by the
$\ell_2$ and $\ell_4$-norms. In particular, in the most technically
involved part of our argument we bound the error from a single 
dimension reduction step as a function of the $\ell_2,\ell_4$ norms of
the test vector. We do this by reducing the problem to that of finding good low-degree polynomial approximators for a product of cosines and then find such low-degree polynomials by careful analytic calculations. This step also critically exploits the fact that when analyzing a single dimension-reduction step, we implicitly have a lot of true randomness in the form of the assignments in the dimension reduced space. 
Finally, we argue separately that the $\ell_2, \ell_4$ norms are not distorted too much by the dimension reduction steps.}
\pnote{Not very satified with this. Verbose, but the novelty of the
  proof does not come across.}


\section{Preliminaries}
\label{sec:prelims}

We start with some notation:
\begin{itemize}
\item For vectors $x \in  \R^n$, let $\nmp{x}$ denote the usual
  $\ell_p$-norms, and let $\nm{x}_0$ denote the size of the support of
  $x$. For a random variable $X$ and $p > 0$, let $\nmp{X} =
  \E[|X|^p]^{1/p}$.
\item For a multi-linear polynomial $Q:\R^n \to \R$, $\nmt{Q}^2$ denotes the sum of squares of coefficients of $Q$ and $\nmo{Q}$ denotes the sum of absolute values of the coefficients.
\item For vectors $u,v \in \R^n$, let $u \star v =(u_iv_i)_{i=1}^n$
  denote the coordinate-wise product.
\item For $v \in \dpm^n$ and $\alpha \in \R$, define
  $\phi_{v,\alpha}(x) = \exp(2 \pi i \alpha (v \cdot x))$.
\item For $v \in\R^n$ and a hash function $h:[n] \to [m]$, define
\begin{align}
\label{eq:hv}
h(v) = \sum_{j=1}^m \|v_{|h^{-1}(j)}\|_2^4
\end{align}  
\item For a hash function $h:[n] \to [m]$, let $A(h) \in \zo^{m \times
  n}$, be the matrix with $A(h)_{ji} = 1$ if and only if $h(i) = j$.
\item For a string $x \in \dpm^n$, let $D(x) \in \R^{n\times n}$ be
  the diagonal matrix formed by  $x$.
\item For two random variables $X,Y$ over a domain $\Omega$, their statistical distance is defined as $\dtv(X,Y) = \max_{A\subseteq \Omega} |\pr[X \in A] - \pr[Y \in A]|$. 
\item Unless otherwise stated $c,C$ denote universal constants.
\end{itemize}
Throughout we assume that $n$ is sufficiently large and that
$\delta,\epsilon>0$ are sufficiently small. All vectors here will be row vectors rather than column vectors.

\begin{defn}
For $n,m,\delta>0$ we say that a family of hash functions $\hh = \{h:[n] \to [m]\}$ is $\delta$-biased if for any $r \leq n$ distinct indices $i_1,i_2,\ldots,i_r \in [n]$ and $j_1,\ldots,j_r \in [m]$,
\[ \pr_{h\in_u \hh} \sbkets{ h(i_1) = j_1 \,\wedge \,h(i_2) = j_2
  \,\wedge\,\cdots\,\wedge h(i_r) = j_r} = \frac{1}{m^r} \pm \delta. \]

We say that such a family is $k$-wise independent if the above holds with $\delta=0$ for all $r\leq k$. We say that a distribution over $\dpm^n$ is $\delta$-biased or $k$-wise independent if the corresponding family of functions $h:[n]\to[2]$ is.
\end{defn}

Such families of functions can be generated using small seeds.
\begin{fact}\label{HashFamilyFact}
For $n,m,k,\delta>0$, there exist explicit $\delta$-biased families of hash functions $h:[n]\to [m]$ that are generated from a seed of length $s=O(\log(n/\delta))$. There are also, explicit $k$-wise independent families that are generated from a seed of length $s=O(k\log(nm))$.

Taking the pointwise sum of such generators modulo $m$ gives a family of hash functions that is both $\delta$-biased and $k$-wise independent generated from a seed of length $s=O(\log(n/\delta)+k\log(nm))$.
\end{fact}

\subsection{Basic Results}

We collect some known results about pseudorandomness and prove some
other technical results that will be used later. We defer all proofs
for this section to Appendix \ref{sec:appendix}.

We will use the following result from \cite{GopalanMRZ13} giving $\prg$s
for signed majorities.
\begin{thm}\cite{GopalanMRZ13}
\label{simpleGeneratorThm}
For $n,\eps>0$ there exists an explicit pseudorandom generator, $Y\in \dpm^n$ generated from a seed of length $s=O(\log(n)+\log^2(1/\epsilon))$ so that for any $v\in\zpm^n$ and $X \in_u \dpm^n$, we have that $\dtv(v\cdot Y,v\cdot X)\leq \epsilon$.
\end{thm}

We shall use $\prg$s for small-space machines or read-once branching programs (ROBP) of Nisan \cite{Nisan92} and Impagliazzo, Nisan and Wigderson \cite{ImpagliazzoNW94}.

\begin{definition}[$(S,D,T)$-ROBP]
An $(S,D,T)$-ROBP $M$ is a layered directed graph with $T+1$ layers and $2^S$ vertices per layer with the following properties.
\begin{itemize}
\item The first layer has a single \emph{start} node and the last layer has two nodes labeled $0,1$ respectively.
\item A vertex $v$ in layer $i$, $0\leq i < T$ has $2^D$ edges to layer $i+1$ each labeled with an element of $\zo^D$.
\end{itemize}
A graph $M$ as above naturally defines a function $M:\left(\zo^D\right)^T \to \zo$ where on input $(z_1,\ldots,z_T) \in \left(\zo^D\right)^T$ one traverses the edges of the graph according to the labels $z_1,\ldots,z_T$ and outputs the label of the final vertex reached.
\end{definition}

\begin{thm}[\cite{Nisan92}, \cite{ImpagliazzoNW94}]\label{th:inwprg}
There exists an explicit $\prg$ $\calG^{INW}:\zo^r \to \left(\zo^D\right)^T$
which $\epsilon$-fools $(S,D,T)$-branching programs and has
seed-length $r = O(D + S \log T + \log(T/\delta) \cdot (\log T))$. 
\end{thm}

We will need to make use of the hypercontractive inequality (see \cite{ODonnell}):
\begin{lem}[Hypercontractivity]\label{hcLem}
Let $x \sim_u \dpm^n$. Then, for a degree $d$ polynomial $Q$ and an even integer $p \geq 2$,
$$\E\sbkets{ Q(x)^p} \leq (p-1)^{p d/2} \cdot \|Q\|_2^p.$$
\end{lem}

\begin{lem}[Hypercontractivity $\delta$-biased]\label{hcBiasedLem}
Let $x \sim \calD$ be drawn from a $\delta$-biased distribution. Then,
for a degree $d$ polynomial $Q$ and an even integer $p \geq 2$, 
\[ \E\sbkets{ Q(x)^p} \leq (p-1)^{p d/2} \cdot \|Q\|_2^p + \|Q\|_1^p \delta.\]
\end{lem}

We will use the following Chernoff-like tail bound for small-bias spaces.
\begin{lem}\label{lm:epschernoff}
  For all $v \in \rn$ with $\nmt{v} = 1$ and $x \sim \calD$ $\epsilon$-biased over $\dpm^n$, and $t \geq 1$
$$\pr\sbkets{\abs{\iprod{v}{x}} > t} \leq 2 \exp(-t^2/4) + \nm{v}_0^{t^2} \cdot \epsilon.$$
\end{lem}

The next two lemmas quantify load-balancing properties of
$\delta$-biased hash functions in terms of the $\ell_p$-norms of
vectors. 

\begin{lem}\label{hashmomentsLem}
Let $p\geq 2$ be an integer. Let $v \in \R^n$ and $\hh = \{h:[n] \to [m]\}$ be either a $\delta$-biased hash family for $\delta>0$ or a $p$-wise independent family for $\delta=0$. Then
$$\E[h(v)^p] \leq O(p)^{2p} \left(\frac{\|v\|_2^4}{m}\right)^p + O(p)^{2p} \|v\|_4^{4p}+ m^p \|v\|_2^{4p}\delta.$$
\end{lem}

\begin{lem}\label{lm:hashing}
For all $v \in \R^n_+$ and $\hh = \{h:[n] \to [m]\}$ a $\delta$-biased family, and $j \in [m]$, and all even $p \geq 2$,
$$\pr\sbkets{ \abs{\nmo{v_{|h^{-1}(j)}} - \nmo{v}/m} \geq t} \leq \frac{O( p)^{p/2} \nmt{v}^p + \nmo{v}^p \delta}{t^p}.$$
\end{lem}


\section{Derandomizing the Chernoff Bounds}

In this section we present a pseudorandom generator that gives
Chernoff-like tail bounds.

\begin{thm}
\label{th:mainChernoff}
For all $\delta > 0$, there exists an explicit generator $\calG:\{0,1\}^r \to
\dpm^n$ with seed-length $r = \tilde{O}(\log(n/\delta))$ such that for
all unit vectors $w \in \rn$, and $t \geq 0$,
\[ \pr_{y \in_u \{0,1\}^r}[ | \iprod{w}{\calG(y)}| \geq t] \leq 4\exp(-t^2/16) +
\delta.\]
\end{thm}

Our construction proceeds in two steps.  We first construct a generator which
has moderate tail bounds but does not match the tail behaviour of
truly random distribution. We then boost the tail bounds to match the
behaviour of truly random distributions using $\prg$s for small-space
machines.

\subsection{Moderately Decaying Tails}
The main result of this section is a generator with the following tail behaviour.

\begin{lem}
\label{lm:chernofffirst}
For $n$ and $\gamma \in(0,1)$, there exists an explicit generator
$\calG':\{0,1\}^{r'} \rgta \dpm^n$ with seed-length $r' =
O(\log(n/\gamma)\log \log(n))$  such that for all unit vectors $w \in \R^n$,
\[ \Pr_{y \in_u \{0,1\}^r}\sbkets{|\iprod{w}{\calG'(y)}| \geq
  (C_1\log(n/\gamma))^{C_2\log\log(n)}} \leq \gamma.\]
\end{lem}

The generator $\calG'$ is recursively defined. We first specify the
{\em one-step} generator $\calG''$ that is used in defining $\calG'$.
Fix $\delta > 0$, $n$. Let $\hh = \{h:[n] \to [m]\}$ be a family of
$\delta$-biased hash functions. Let $\calD$ be a $\delta$-biased
distribution over $\dpm^n$. The generator $\calG''$ takes as input a hash
function $h \in \hh$, $x \in \calD$ and $z \in_u \dpm^m$, the output
is
\[ \calG''(h,x,z) = z A(h) D(x).\]
Thus we have for any $i \in [n]$,
\[ \calG''(h,x,z)_j = z_{h(i)}x_i.\]
Thus the generator $\calG''$ starts with the $\delta$-biased string $x \in
\calD$ as output, hashes the coordinates into $[m]$ bins and
flips the signs of all coordinates in each bin by picking a uniformly
random independent bit for each bin. This takes $O(\log(n/\delta) + m)$ random bits.

The generator $\calG'$ is obtained by taking $m \approx
\sqrt{n}$ and then recursively using $\calG''$  to generate $z \in
\pmo^n$. The base case of the recursion is reached when $m =
O(\log(n/\delta))$ at which point we use a truly random string
$z$. This requires $k \leq \log\log(n)$ stages of recursion, so that
the seed length is $O(\log(n/\delta)\log\log(n))$. Unrolling the
recursion, we see that if we set $n_\ell =n^{2^{-\ell}}$ for $\ell
\in\{0,\ldots, k\}$ then $\calG'$ takes as input two sequences:
\begin{itemize}
\item A sequence of hash functions $h^1,\ldots,h^k$ where
  $h^\ell:[n_{\ell-1}] \rgta [n_\ell]$ is drawn from a $\delta$-biased family
  of hash functions.
\item A sequence of strings $x^1,\ldots,x^k$ where $x^\ell \in \pmo^{n_\ell}$ is drawn
  from a $\delta$-biased distribution.
\end{itemize}
For each coordinate $i$, consider the sequence $\{i_\ell \in
n_\ell\}_{\ell =0}^k$ obtained by successively applying the hash
functions:
\[ i_0 =i, \ \ i_\ell = h^\ell(i_{\ell -1}) \ \text{for} \ \ell \geq 1.\]
Then we have
\[ \calG'(h^1,\ldots,h^k,x^1,\ldots,x^k) = \prod_{\ell=1}^kx^\ell_{i_\ell}. \]

The analysis of $\calG'$ proceeds step by step and each step reduces
to analyzing $\calG''$. Note that
\begin{align*}
\iprod{w}{\calG''(h,x,z)} =\iprod{w}{zA(h)D(x)} =
\iprod{A(h)D(x)w^T}{z}.
\end{align*}
We can view $A(h)D(x)w^T$ as projection of $w \in \R^n$ down to
$\R^m$ where we first hash coordinates into buckets, and then sum the
coordinates in a bucket with signs given by $x$. The next lemma
saying that the transformation $A(h)D(x)$ is unlikely to  stretch
Euclidean norms too much serves as the base case for the
recursion.


\begin{lem}
\label{lem:v-technical}
Let $n \geq 1$ and $m = \sqrt{n}$ and $\delta < 1/10n^2$. Let $\calD$ be a
$\delta$-biased distribution over $\dpm^n$ and $\hh = \{h:[n]
\to [m]\}$ be a $\delta$-biased hash family.
There exists a constant $C$ such that for all  unit vectors $w \in \R^n$,
\begin{align*}
\Pr_{x \in \calD, h \in \hh}\sbkets{\nmt{A(h)D(x)w^T} \geq
  C(\log\log(n))\log(1/\delta)^{3/4}} \leq  3(\log\log(n))m\sqrt{\delta}
\end{align*}
\end{lem}

We prove this lemma by decomposing the vector $w$ across various weight scales.
Fix a unit vector $w \in \R^n$. Without loss of generality, we ignore
all coordinates $i$ where $|w_i| \leq 1/n$ as they can only
effect the $\ell_2$-norm by at most $1$. For $\ell \in
\{1,\ldots,\log\log n\}$, define $w(\ell) \in \R^n$ as
\begin{align*}
w(\ell)_i  = \begin{cases} w_i & \text{if} \ |w_i| \in \left(\frac{1}{2^{2^{\ell}}},
  \frac{1}{2^{2^{\ell-1}}}\right] \\
0 & \text{otherwise}.
\end{cases}
\end{align*}
Thus $w(\ell)$ picks out the entries in the $\ell^{th}$ weight scale.
In addition, define $w(0)$ to consist of entries that lie in the
interval $(1/2,1]$. We will show that for every $\ell$, the bound
\[\nmt{A(h) D(x) w(\ell)^T}^2 \leq O(1) \log(1/\delta)^{1.5}\]
holds with high (inverse polynomial) probability. Here we tailor the
amount of independence we use in the argument to the scale, in a
manner similar to \cite{CelisRSW13,GopalanMRTV12}. Once this is done,
Lemma \ref{lem:v-technical} follows by the triangle inequality.

We start with a simple bound which suffices for small constant $\ell$.

\begin{lem}
\label{lem:small-ell}
For all $\ell$, we have
\begin{align*}
\nm{A(h)D(x)w(\ell)^T}_2 \leq 2^{2^\ell}.
\end{align*}
\end{lem}
\begin{proof}
Observe that
\begin{align*}
\nm{A(h)D(x)w(\ell)^T}_\infty \leq \nm{w(\ell)}_1 \\
\nm{A(h)D(x)w(\ell)^T}_1 \leq \nm{w(\ell)}_1
\end{align*}
hence by Holder's inequality
\begin{align*}
\nmt{A(h)D(x)w(\ell)^T} \leq \nm{w(\ell)}_1.
\end{align*}
Since $\nmt{w(\ell)} \leq 1$ and every  non-zero entry is at least
$2^{-2^\ell}$ we have $\nm{w(\ell)}_1 \leq 2^{2^\ell}$ and hence
\begin{align*}
\nm{A(h)D(x)w(\ell)^T}_2 \leq 2^{2^\ell}.
\end{align*}
\end{proof}

Given this lemma, we can assume that $\ell$ is a sufficiently large constant.
We show that the weight vector is hashed fairly
regularly with high probability over the choice of $h \in \hh$, where the
regularity is measured by $h(v)$.

\begin{lem}
\label{lem:good-h}
Fix $\ell \geq 2$. Then
\begin{align}
\label{eq:good-h}
\Pr_{h \in \hh}\sbkets{h(w(\ell)) \leq 2 C_4 \frac{\sqrt{\log(1/\delta)}}{
    2^{2^{\ell-2}}}} \geq 1 - 2m\sqrt{\delta}.
\end{align}
\end{lem}
\begin{proof}
By Lemma \ref{lm:hashing} applied to the vector
$(w(\ell)_i^2)_{i=1}^n$, there is a constant $C_4$ such that for even $q \geq 2$
\[ \pr_{h \in \hh} \sbkets{h(w(\ell)) \geq 1/m + t} \leq m\left(\left(\frac{C_4 \sqrt{q}
  \nm{w(\ell)}_4^2}{t}\right)^q + \delta\left(\frac{\nmt{w(\ell)}^2}{t}\right)^q \right).\]
Plugging in the bounds
\begin{align*}
\nm{w(\ell)}_4^2  \leq \nm{w(\ell)}_\infty \leq 2^{-2^{\ell-1}}, \ \nmt{w(\ell)} \leq 1
\end{align*}
we get
\[ \pr_{h \in \hh} \sbkets{h(w(\ell)) \geq 1/m + t} \leq
m\left(\left(\frac{C_4 \sqrt{q}}{2^{2^{\ell -1}}t}\right)^q + \delta\left(\frac{1}{t}\right)^q \right).\]

Therefore, taking
\[ q = \frac{\log(1/\delta)}{2^{\ell-1}}, \ \ t = \frac{C_4\sqrt{q}}{2^{2^{\ell-2}}} \]
in the above equation, we get
\begin{align*}
\frac{C_4\sqrt{q}\nm{w(\ell)}_4^2}{t} & \leq
\frac{2^{2^{\ell-2}}}{2^{2^{\ell-1}}} \leq\frac{1}{2^{2^{\ell -2}}},\\
\left(\frac{C_4\sqrt{q}\nm{w(\ell)}_4^2}{t}\right)^q & \leq
\frac{1}{2^{2^{\ell-2}\log(1/\delta)/2^{\ell-1}}} \leq \sqrt{\delta}.\\
\left(\frac{1}{t}\right)^q & \leq  (2^{2^{\ell
    -2}})^{\log(1/\delta)/2^{\ell-1}}  = \frac{1}{\sqrt{\delta}}.
\end{align*}
hence
\[
\pr_{h \in \hh} \sbkets{h(w(\ell)) \geq 1/m + C_4 \frac{\sqrt{\log(1/\delta)}}{ 2^{2^{\ell - 2}}}}
\leq 2 m \sqrt{\delta}.
\]
Hence with probability $1 - 2m\sqrt{\delta}$ over the choice of $h$,
we have
\begin{align*}
h(w(\ell)) \leq \frac{1}{m} + \frac{C_4\sqrt{\log(1/\delta)}}{ 2^{2^{\ell-2}}} \leq 2 C_4 \frac{\sqrt{\log(1/\delta)}}{ 2^{2^{\ell-2}}}
\end{align*}
since $1/m = 1/\sqrt{n} \leq 2^{-2^{\ell -2}}$.
\end{proof}

Conditioned on the hash function $h$ being good (i.e., satisfying the condition of the previous lemma), we will show that
$\nmt{A(h)D(x)w(\ell)^T}^2$ is small with high-probability.

\begin{lem}
\label{lem:good-x}
Fix $\ell \geq 2$ and assume that $h$ is such that the event described in Equation \eqref{eq:good-h}
holds. There exists a constant $C_6$ such that
\begin{align}
\label{eq:good-x}
\Pr_{x \in \calD} \sbkets{\nmt{A(h) D(x) w(\ell)^T} \leq
C_6\log\left(\frac{1}{\delta}\right)^{5/8}} \geq 1- \sqrt{\delta} -\delta^{16}.
\end{align}
\end{lem}
\begin{proof}
We will show that $\nmt{A(h)D(x)w(\ell)^T}^2$ is concentrated around
its mean (which is $\nmt{w(\ell)}^2$) by bounding its moments. The
deviation is given by the polynomial
\begin{align*}
Q_\ell(x) & =  \nmt{A(h)D(x)w(\ell)^T}^2 -\nmt{w(\ell)}^2\\
& = \sum_{j \in [m]} \bkets{\sum_{i \in h^{-1}(j)} w(\ell)_i x_i}^2 -\nmt{w(\ell)}^2\\
& =  \sum_{j \in [m]} \sum_{i_1 \neq i_2 \in h^{-1}(j)} w(\ell)_{i_1}w(\ell)_{i_2} x_{i_1} x_{i_2}
\end{align*}
We have
\begin{align*}
\E_{x \in_u \pmo^n}[Q_\ell(x)^2] & = \sum_{j=1}^m \sum_{i_1 \neq i_2 \in h^{-1}(j)}
  (w(\ell)_{i_1})^2 (w(\ell)_{i_2})^2\\
& \leq \sum_{j=1}^m  \|w(\ell)_{|h^{-1}(j)}\|_2^4 \\
&= h(w(\ell)).\\
\nmo{Q_\ell} & = \sum_{j \in [m]} \sum_{i_1 \neq i_2 \in h^{-1}(j)} |w(\ell)_{i_1}w(\ell)_{i_2}| \\
& \leq \sum_{j \in m} \nmo{w(\ell)_{|h^{-1}(j)}}^2 \\
& \leq \nmo{w(\ell)}^2\\
& \leq \nm{w(\ell)}_0.
\end{align*}

By Lemma \ref{hcBiasedLem} applied to $Q$ with $d=2$, there exists a
constant $C_3$ so that for all even $p \geq 2$,
\begin{equation}\label{eq:firststep}
  \E_{x \in \calD}[Q(x)^p] \leq \left(C_3 p \sqrt{h(w(\ell))}\right)^p
  + \nm{w(\ell)}_0^{p} \delta.
\end{equation}

We bound $h(w_\ell)$ using Equation \eqref{eq:good-h}. We also
have $\nm{w(\ell)}_0 \leq 2^{2^\ell}$. Plugging these into Equation \eqref{eq:firststep},
\[
\E_{x \in \calD}[Q(x)^p] \leq
\frac{C_5^pp^p\log(1/\delta)^{p/4}}{2^{2^{\ell -3}p}} + 2^{2^\ell p}\delta
\]
Now setting
\[ p = \frac{\log(1/\delta)}{2^{\ell+1}}, \ \theta =
C_5\log(1/\delta)^{5/4} \]
and using Markov' inequality gives
\[
\Pr_{x \in \calD}[Q(x) \geq \theta] \leq
\left(\frac{C_5p\log(1/\delta)^{1/4}}{2^{2^{\ell -3}}\theta}\right)^p
+  \left(\frac{2^{2^\ell}}{\theta}\right)^p\delta
\]
To bound the first term, note that
\[
\left(\frac{C_5p\log(1/\delta)^{1/4}}{2^{2^{\ell -3}}\theta}\right)^p \leq
  \left(\frac{1}{2^{2^{\ell-3}}}\right)^{\log(1/\delta)/2^{\ell+1}}
  \leq \delta^{16}.
\]
For the second term, note that since $\theta \geq  1$,
\[
\left(\frac{2^{2^\ell}}{\theta}\right)^p\delta \leq 2^{2^\ell
  \log(1/\delta)/2^{\ell +1}}\delta \leq \sqrt{\delta}
\]

Therefore, except with probability at least $\sqrt{\delta}+ \delta^{16}$ we have
\[\nmt{A(h) D(x) w(\ell)^T}^2 \leq \nmt{w_\ell}^2 + C_5\log(1/\delta)^{5/4}\]
hence
\[\nmt{A(h) D(x) w(\ell)^T} \leq C_6\log(1/\delta)^{5/8}\]
\end{proof}

We now finish the proof of Lemma \ref{lem:v-technical}.

\begin{proof}[Proof of Lemma \ref{lem:v-technical}]
Note that
\[ w = \sum_{i=0}^{\log\log(n)}w(\ell).\]
We will assume that $h$ and $x$ are chosen so that the conditions in Equations
\eqref{eq:good-h} and \eqref{eq:good-x} hold for all $\ell$. By the
union bound, this happens except with probability
\[
\log\log(n)(2m\sqrt{\delta} + \sqrt{\delta} + \delta^{16}) < 3\log\log(n)m\sqrt{\delta}.
\]
In which case, we have
\begin{align*}
\nmt{A(h)D(x) w} & = \nmt{\sum_{\ell=0}^{\log\log(n)} A(h) D(x) w(\ell)^T} \\
& \leq \sum_{\ell=0}^{\log\log(n)} \nmt{A(h) D(x) w(\ell)^T}\\
& \leq C_6\log\log(n) \log(1/\delta)^{5/8}\\
&\leq C_6\log(1/\delta)^{3/4}.
\end{align*}
\end{proof}

We now prove the main lemma of this section:
\begin{proof}[Proof of Lemma \ref{lm:chernofffirst}]
Let $k = \log\log(n)$ be the number of recursive stages.
Let $w$ be a unit vector. Given  $h^1,\ldots,h^k$ and $x^1,\ldots,x^k$, we have
\begin{align*}
\iprod{w}{\calG'(h^1,\ldots,h^k,x^1,\ldots,x^k)}  = \prod_{\ell=1}^kA(h^\ell)D(x^\ell)w^T
\end{align*}

Let $C_8n\sqrt{\delta} = \gamma$, so that $\delta =
\Omega(\gamma^2/n^2)$. Note that $\log(1/\delta)\gg \log\log(n)^4$.

By applying Lemma \ref{lem:v-technical} inductively and using the
union bound, except with probability
\[ 3(\log\log(n))^2m\sqrt{\delta} \leq C_8n\sqrt{\delta} \]
we have that for every $i \leq k$
\begin{align*}
\nmt{\prod_{\ell =1}^iA(h^\ell)D(x^\ell)w^T} \leq (C_6\log\log(n)\log(1/\delta)^{3/4})^i \leq \log(1/\delta)^{i/2}
\end{align*}
and hence
\begin{align*}
|\iprod{w}{\calG'(h^1,\ldots,h^k,x^1,\ldots,x^k)}|  \leq
\log(1/\delta)^{\log\log(n)/2}.
\end{align*}

Thus, with probability $1 - \gamma$, the
deviation is bounded by
\[d(n,\gamma) := (C_1\log(n/\gamma))^{C_2\log\log(n)}\]
and the seedlength of this generator is
\[r' = O(\log(n/\delta)\log\log(n)) = O(\log(n/\gamma)\log\log(n)).\]
\end{proof}

\subsection{Getting sub-Gaussian tail bounds}

The generator $\calG'$ gives a tail probability of
$1-\gamma$ pseudorandomly for $d(n,\gamma)$ standard deviations. We now boost this to obtain sub-Gaussian tails by starting with independent copies of
$\calG'$ and then reuse the seeds for  using a $\prg$ for space bounded computations.

\newcommand{\gmain}{\bar{\calG}}

We will make some added assumptions about $\calG'$:
\begin{itemize}
\item The output is $\eps$-biased for some $\eps \ll \gamma$. We can ensure this by xor-ing the output with an $\eps$-biased string.

\item The distribution is symmetric: for every $x$,  $\Pr_y[\calG'(y)
  =x] = \Pr_y[\calG'(y) = -x]$. We ensure this by  outputting either
  $\calG'(y)$ or $-\calG'(y)$ with probability $1/2$.
\end{itemize}

Let $D_1,D_2$ be constants such that
\[ m = (D_1\log(n/\gamma))^{D_2\log\log(n)} >  10 d(n,\gamma)^2\log(1/\gamma). \]
Note that for $\gamma = 1/\poly(n)$, $\log(m) = O(\log\log(n)^2)$.

Let $\hh = \{h:[n] \to [m]\}$ be a family of $\gamma$-biased hash
functions. Define a new generator $\gmain:\bkets{\zo^r}^m \times \hh
\to \dpm^n$ as follows:
\begin{equation}
  \label{eq:genmain}
 \gmain(z_1,\ldots,z_T,h)_i = \calG'(z_j)_i, \text{ if $h(i) = j$}.
\end{equation}

The seed-length of the generator is $\bar{r} = \log(n/\delta) + m \cdot r'$ which
we will later improve to $\log(n/\delta) + r' + \log (n/\gamma)\log(m)$
using $\prg$s for space bounded computations.

The following claim characterizes the tail behaviour of the output of $\gmain$.
\begin{lem}\label{lm:chernoffsecond}
Let $0 < \eps < \gamma \leq 1/n^3$. For all unit vectors $w \in \R^n$, the
generator $\gmain$ satisfies
\[
\pr_{y \in \{0,1\}^{\bar{r}}}\sbkets{\left|\iprod{w}{\gmain(y)}\right|
  \geq  t} \leq 4(\exp(-t^2/16) + m\sqrt{\gamma} +
(1/\gamma)^{4\log(m)} \eps).
\]
\end{lem}
\begin{proof}
Note that it suffices to prove the claim for $t \leq
2\sqrt{\log(1/\gamma)}$ since tail probabilities can only decrease
with $t$ and beyond this value, the tail bound is dominated by the
additive terms.

Fix a unit vector $w \in \R^n$. Let
\begin{align}
\label{eq:set-beta}
\beta = \frac{1}{m^2\sqrt{\log(1/\gamma)}}
\end{align}
and define $u, v \in \R^n$ to consist of the heavy and light indices respectively
\begin{align*}
u_i = \begin{cases} w_i & \text{if} \ |w_i| \geq \beta\\
0 &\text{otherwise.}
\end{cases}, \ \
v_i = \begin{cases} w_i & \text{if} \ |w_i| < \beta\\
0 &\text{otherwise.}
\end{cases}.
\end{align*}
Since $w =u +v$, it suffices to bound the probability that 
either of $|\iprod{u}{\gmain(y)}|$ and $|\iprod{v}{\gmain{y}}|$ exceed $t/2$.
We will consider $u$ first. Note that
\[ \nm{u}_0 \leq \frac{1}{\beta^2} \leq m^4 \log\left(\frac{1}{\gamma}\right).\]
Since $\gmain(y)$ is $\eps$-biased, by Lemma \ref{lm:epschernoff}
applied for $t/2 \leq \sqrt{\log(1/\gamma)}$,
\begin{align}
\pr\sbkets{\abs{\iprod{u}{\gmain(y)}} > t/2} & \leq 2 \exp(-t^2/16) +
\nm{u}_0^{\log(1/\gamma)} \eps \notag\\
& \leq 2 \exp(-t^2/16) + (m\log(1/\gamma))^{\log(1/\gamma)}  \eps.
\label{eq:finalmain2}
\end{align}

We will show a tail bound for $\iprod{v}{\gmain(y)}$ by bounding its
higher order moments. Fix a hash function $h \in_u \hh$ and for $j \in \{1
,\ldots,m\}$, let
\[ Z_j = \iprod{v_{|h^{-1}(j)}}{\calG'(z_j)}.\]
Note that  the random variables $Z_j$ are independent of one another,
and
\[ \iprod{v}{\gmain(y)} = \sum_{j=1}^m Z_j.\]
We use Lemma \ref{lm:hashing} to bound
$\nmt{v_{|h^{-1}(j)}}$. We defer the proof of the following technical
lemma.

\begin{lem}
\label{lem:tech1}
With probability $1 - 2m\sqrt{\gamma}$, for all $j \in [m]$ we have
\begin{align}
\label{eq:bound-h}
\nmt{v_{|h^{-1}(j)}}^2 \leq  \frac{2}{m}
\end{align}
\end{lem}
We condition on the hash function $h$ satisfying
Equation \eqref{eq:bound-h}, and call this event $A$.

Recall that $Z_j = \iprod{v_{|h^{-1}(j)}}{\calG'(z_j)}$.
By Lemma \ref{lm:chernofffirst}, with probability $1 -\gamma$ over
$z_j$, we have the bound
\begin{align}
\label{eq:bound-zj}
|Z_j| \leq d(n,\gamma) \nmt{v_{|h^{-1}(j)}} \leq
\frac{\sqrt{2}d(n,\gamma)}{\sqrt{m}} \leq
\frac{1}{2\sqrt{\log(1/\gamma)}} := M
\end{align}
where the last inequality is by the choice of $m$.
By the union bound, Equation \eqref{eq:bound-zj} holds with
probability at least $1 - m\gamma$ over $z_1,\ldots,z_m$, for all $j
\in [m]$. We further condition on the event $|Z_j| \leq M$ for all $j
\in [m]$ which we denote by $B$.

Conditioning on a high probability event preserves the small-bias
property of $G(z_j)$'s up to a small additive error. In particular,
conditioned on the event $B$, $\calG'(z_j)$ is $(\eps + m\gamma)$-biased. Since
$\nm{v_{|h^{-1}(j)}}_1 \leq \sqrt{n}$ we have
\begin{align*}
\sum_{j=1}^m \E[Z_j^2|B] & \leq \sum_{j=1}^m (\nmt{v_{|h^{-1}(j)}}^2 + n
  (\eps+m \gamma)) \\
& \leq 1 + nm(\eps + m\gamma)\\
& \leq 2.
\end{align*}
Further, since $\calG'(z_j)$ is symmetric, it continues to be
symmetric after we condition on $B$ (which is a symmetric event in
$\calG'(z_j))$.

Since $t/2 \leq \sqrt{\log(1/\gamma)}$, and $M =
1/2\sqrt{\log(1/\gamma)}$ we have $Mt\leq 1/2$.
We now apply Bernstein's inequality \cite{Fellerbook} to the random variables
$\{Z_j|B\}_{j=1}^m$ which are mean zero and are bounded by $M$ to get
\begin{align*}
 \pr\sbkets{\abs{\sum_{j=1}^m Z_j|B} > t/2} &\leq 2
 \exp\bkets{-\frac{t^2}{4(\sum_j \nmt{Z_j|B}^2 + M t/3)}}\\
& \leq 2 \exp\bkets{-\frac{t^2}{28/3}} \\
& \leq 2 \exp\bkets{-\frac{t^2}{16}}.
\end{align*}
Combining the above arguments, we get that
\begin{equation}
\label{eq:finalmain1}
  \pr\sbkets{\abs{\iprod{v}{\calG'(y)}} > t/2} \leq 2 \exp(-t^2/16) +  2m\sqrt{\gamma}.
\end{equation}

The claim now follows from combining Equations \eqref{eq:finalmain1} and \eqref{eq:finalmain2}
\end{proof}

\begin{proof}[Proof of Lemma \ref{lem:tech1}]
Note that
\[\nmo{v \star v} \leq 1, \ \nmt{v \star v} = \nm{v}_4^2 \leq \beta.\]
Therefore, by setting
\[ p = \frac{\log(1/\gamma)}{\log(1/\beta)}, \ t = D_6 \sqrt{\beta p}\]
in Lemma \ref{lm:hashing} we get
\begin{align*}
\pr_{h \in \hh}\sbkets{\nmt{v_{|h^{-1}(j)}}^2 \geq \frac{1}{m} + D_6
  \sqrt{\beta\frac{\log(1/\gamma)}{\log(1/\beta)}}}  & \leq
  \frac{D_6^p p^p \beta^p + \gamma}{(D_6 \sqrt{\beta p})^p} \leq 2 \sqrt{\gamma}.
\end{align*}

By a union bound, with probability at least $1 - 2m \sqrt{\gamma}$
over $h$, for all $j \in [m]$,
\begin{align*}
\nmt{v_{|h^{-1}(j)}}^2 & \leq \frac{1}{m} +
\sqrt{\beta\frac{\log(1/\gamma)}{\log(1/\beta)}}\\
& \leq \frac{2}{m}
\end{align*}
where the last inequality follows from our choice of $\beta$ in
Equation \eqref{eq:set-beta}.
\end{proof}

\subsection{Putting things together}
We are now ready to prove Theorem \ref{th:mainChernoff}.

\begin{proof}[Proof of Theorem \ref{th:mainChernoff}]
Let $\gmain:\bkets{\{0,1\}^{r'}}^m \times \hh \to \dpm^n$ be the
generator as in Lemma \ref{lm:chernoffsecond}. Let $\delta$ be the
final additive error desired. Set
\begin{align*}
\gamma & = \frac{\delta^2}{(D_5\log(n/\delta))^{D_6\log \log(n)}} \leq \frac{\delta^2}{50m^2}\\
\eps  & = \left(\frac{\delta}{n}\right)^{D_7\log\log(n/\delta)^3} \leq
\frac{\delta}{20 \gamma^{4 \log(m)}}.
\end{align*}
It can be verified that with these parameter seedings, the error
probability in Lemma \ref{lm:chernoffsecond} is at most $4e^{-t^2/16} + \delta$,
and the seed-length of $\gmain$ is $\log(n/\gamma) + m r'$, where $r' = O(\log(n/\delta) \cdot (\log \log (n/\delta))^3)$. 

Observe that once we fix the has function $h$, the inner product
$\iprod{w}{\calG'(z_1,\ldots,z_m,h)}$ can be computed by a $(S, r',m)$-ROBP where $S = O(\log n)$ which reads one $z_1,z_2,\ldots,z_m$ in order. The reason is that we can round each weight $w_i$ up to a multiple of $1/n^2$. This can only increase $|\iprod{w}{z}|$ by $1/n$
for any $z \in \pmo^n$. This also ensures that $\iprod{w}{z}$ lies in the interval $[-\sqrt{n},\sqrt{n}]$ and that it is a multiple of $1/n^2$ and thus can be computed with $O(\log n)$-bits of precision.

Now, let $\calG^{INW}:\{0,1\}^{r_s} \to \bkets{\{0,1\}^{r'}}^m$ be a
generator fooling $O(S,r',m)$-ROBP with error $\delta$. By Theorem
\ref{th:inwprg}, there exist such generators with seed-length 
$$r_s = O(r' + (\log n) (\log m) + \log(m/\delta)\cdot (\log m)) =
O((\log (n/\delta))(\log\log(n/\delta))^3).$$ 

Now, if we define our final generator $\calG^f:\hh \times \{0,1\}^{r_s} \to \pmo^n$ by
$$\calG^f(h,z) = \calG'(\calG^{INW}(z),h).$$
From the above arguments it follows that the output of $\calG^f$ only
has an additional $\delta$ error compared to $\calG'$. The theorem now
follows from the above bound on seed-length.  
\end{proof}


\section{A $\prg$ for \rom}

In this section we construct generator to fool \rom to polynomial
error with seedlength $\tilde{O}(\log n)$ proving Theorem
\ref{th:mainintrohs}. As the generator and its analysis is quite
technical, we first give a high-level description at the
risk of repeating parts of Section \ref{sec:overview}.   
 
\subsection*{Proof overview}
For simplicity, in this discussion let us fix a test vector $v \in
\zpm^n$ and error $\epsilon = 1/\poly(n)$. We start by noting that it
suffices to design a PRG $G:\zo^r \to \dpm^n$ such that
$\dtv(\dotp{v}{G(y)}, \dotp{v}{X}) \ll 1/\poly(n)$ where $y \in_u
\zo^r$ and $X \in_u \dpm^n$. In the following let $X \in_u \dpm^n$ and
$Y \sim G(y)$, where $y \in_u \zo^r$ be the output of the desired
generator. 

The starting point of our analysis and construction is to note that showing closeness in statistical distance for discrete random variables is equivalent to showing that the \emph{Fourier transforms} of the random variables are close. This will allow us to use various analytic tools. Concretely, we shall use the following elementary fact about the discrete Fourier transform.
\begin{claim}\label{clm:introfouriertotv}  
  Let $Z_1, Z_2$ be two discrete random variables with support sizes at most $B$. Then,
$$\dtv(Z_1,Z_2) \leq \sqrt{2 B} \cdot \max_{\alpha \in \R} \left|\E[\exp(2\pi i \alpha Z_1)] - \E[\exp(2\pi i \alpha Z_2)]\right|.$$
\end{claim}
\begin{proof}
Note that the distribution $Z_1-Z_2$ is supported on at most $2B$ points. Therefore,
$$
\dtv(Z_1,Z_2) = \|Z_1-Z_2\|_1 \leq \sqrt{2B}\|Z_1-Z_2\|_2.
$$
On the other hand, the Plancherel identity implies that
$$
\|Z_1-Z_2\|_2 \leq \max_{\alpha \in \R} \left|\E[\exp(2\pi i \alpha Z_1)] - \E[\exp(2\pi i \alpha Z_2)]\right|.
$$
This completes the proof.
\end{proof}

Henceforth, we will focus on designing a generator so as to fool the test function $\exp(2 \pi i \alpha \dotp{v}{x}) \equiv \phi_{v,\alpha}(x)$. To do so, we will consider two cases based on how large $\alpha \in [0,1]$ is. The two cases we consider capture the shift in the behaviour of $\E[\exp(2\pi\alpha(v \cdot X))]$ - the ``$\alpha$-th Fourier coefficient''. We can combine the generators for the two cases easily at the end. 

\subsubsection*{Large $\alpha$: $\alpha \gg (\log n)^{O(1)}/\sqrt{\nm{v}_0}$} Roughly speaking, the reason for considering this threshold is that all values of $\alpha$ greater than this value yield similar Fourier coefficients: $\left|\E[\phi_{v,\alpha}(X)]\right| \ll 1/\poly(n)$ for $\alpha$ in this range. Thus, it suffices to ensure that $\E[\phi_{v,\alpha}(Y)]$ is small. We achieve this by exhibiting a way to ``amplify'' the error, i.e., go from fooling $\phi_{v,\alpha}$ with constant error to fooling them with polynomially small error at the expense of a $O(\log \log n)$ factor in seed-length. We then instantiate this amplification procedure with the generator of Gopalan, Meka, Reingold, Zuckerman \cite{GopalanMRZ13} which requires seed-length $O(\log n)$ to fool such test functions ($\phi_{v,\alpha}(\;)$) with constant error. We leave the details of the amplification procedure to the corresponding section.

\subsubsection*{Small $\alpha$: $\alpha \ll (\log n)^{O(1)}/\sqrt{\nm{v}_0}$}
This is the harder of the two cases and the core of our construction and analysis. 
The generator we use is essentially the same as the one based on iterative dimension reduction used in derandomizng the Chernoff bound. The main difference will be that instead of using small-bias spaces in each dimension reduction step we use $k$-wise independent spaces for suitable $k$. However, the analysis is quite different and requires several new analytic tools. 

We next formally describe our generator for handling this case. Let $n,\delta>0$. Let $C$ be a sufficiently large constant. We define a generator as follows. Let $n=n_1>n_2>\ldots > n_t$ so that $n_{i+1}=n_i^{1/2}+O(1)$ and $\log^{2C}(n/\delta)\geq n_t \geq \log^C(n/\delta)$. Note that this implies that $t=O(\log\log(n))$. For $1\leq i < t$, let $\hh_i = \{h:[n_i] \to [n_{i+1}]\}$ be a family of $\frac{C\log(n/\delta)}{\log(n_i)}$-wise independent hash functions. Let $h_i \in_u \hh_i$. Let $Z_i$ be a random element of $\dpm^{n_i}$ chosen from a distribution that is both $(\delta/n)^C$-biased and and $\frac{C\log(n/\delta)}{\log(n_i)}$-wise independent. Finally, let $Z$ be a random variable in $\dpm^{n_t}$ be chosen to fool weight at most $n$ halfspaces to variational distance $\delta/n$ as described in Theorem \ref{simpleGeneratorThm}. We define our random variable $Y\in\dpm^n$ to be
\begin{equation}\label{LargeAlphaEqn}
Y = Z A(h_{t-1})D(Z_{t-1})A(h_{t-2})D(Z_{t-2})\cdots A(h_1)D(Z_1).
\end{equation}

Informally, this generator begins with the string $Z_1$, then uses $h_1$ to divide the coordinates into $n_2$ bins and then for each bin multiplies the elements in this bin by a random sign, these $n_2$ signs being chosen recursively by a similar generator, until at the final level they are picked using the generator from Theorem \ref{simpleGeneratorThm} instead.

It is easy to see from Theorem \ref{simpleGeneratorThm} and Fact \ref{HashFamilyFact} that the random variable $Y$ can be produced from a random seed of length $s=O(\log(n/\delta)\log\log(n/\delta))$. We also claim that it fools $\phi_{v,\alpha}$ for $|\alpha| \leq \log^3(1/\delta)/\|v\|_2$. This in turn implies our claimed pseudorandomness for halfspaces in lieu of Claim \ref{clm:introfouriertotv}.  

To analyze the generator we shall use a hybrid argument to exploit the recursive nature of the generator. To this end, for $1 \leq i < t$, let $X_i \in_u \dpm^{n_i}$ and define
\begin{equation}\label{YiEqn}
Y_i := X_iA(h_{i-1})D(Z_{i-1})\cdots A(h_1)D(Z_1)
\end{equation}
(note that $Y_1=X_1$) and let $Y_t=Y$.

The crux of the analysis is then in showing the following claim analyzing a single dimension reduction step: for $1 \leq i \leq t$ and $\alpha \leq \log^3(1/\delta)/\nmt{v}$,
$$\left | \E[\phi_{v,\alpha}(Y_i)] - \E[\phi_{v,\alpha}(Y_{i+1})]\right| \leq \delta/n.$$
If we let $v_0 = v$ and $v_i = A(h_{i-1})D(Z_{i-1}\cdots A(h_1)D(Z_1)v$, then the above claim amounts to bounding
\begin{equation}
  \label{eq:overview2}
\left | \E[\phi_{v_i,\alpha}(X_i)] - \E[\phi_{v_i,\alpha}(X_{i-1}A(h_i)D(Z_i))]\right|.  
\end{equation}

Thus, intuitively, we need to argue that a single step of dimension reduction (i.e., applying $A(h_i)D(Z_i)$) does not cause too much error. Ideally, we would have liked to make such a claim for all test functions of the form $\phi_{w,\alpha}$; this turns out to be false. What remains true however is that a single dimension reduction step \emph{fools} test functions of the form $\phi_{w,\alpha}$ when the test vector $w \in \R^{n_i}$ is \emph{sufficiently well-spread out} (as measured by the $\ell_2,\ell_4$-norms of $w$) and $\alpha$ is not too large. In particular, in the most technically intensive part of our argument we bound the error from the above step as a function of the $\ell_2,\ell_4$ norms of the vector $v_i$. We then argue separately that the $\ell_2, \ell_4$ norms of the test vector $v$ are close to their true values under the above transformations.

In order to analyze expectations as in Equation \ref{eq:overview2}, it is critical to note that $X_{i-1}$ is uniformly distributed. This implies (for fixed $h_i$) that the given expectation over $X_{i-1}$ is a product of cosines of linear functions of $Z_i$. We take advantage of the fact that cosine is a smooth function of its input, allowing us to approximate this product by a Taylor polynomial. If $\alpha$ is sufficiently small, the higher order terms will be small enough to ignore, and therefore the limited independence of $Z_i$ will be sufficient to guarantee the desired approximation.
\subsection{Generator for large $\alpha$}
We now develop a generator that works when $\alpha$ is large, in particular, we prove:
\begin{prop}\label{genLargeProp}
There exists an explicit generator $\glarge:\zo^r \to \dpm^n$ with seed-length $r = O((\log(n/\epsilon)) (\log \log (n)))$ such that the following holds. For all $v \in \zpm^n$, $\alpha \in (-1/4,1/4)$ with $\alpha \geq \log^3 (1/\epsilon)/\nm{v}_2$,
\[ \left|\E_{y \in_u \zo^r}\sbkets{\phi_{v,\alpha}(\glarge(y))} - \E_{X \in_u \dpm^n}\sbkets{\phi_{v,\alpha}(X)}\right| \leq \epsilon.\]
\end{prop}

\subsubsection{Spreading hashes}

In order to prove Proposition \ref{genLargeProp} we will need to study a certain property of hash families.

\begin{definition}
 A family of hash functions $\hh = \{h:[n] \to [m]\}$ is said to be $(k,\ell,\epsilon)$-spreading if the following holds: for every $I \subseteq [n]$ with $|I| \geq k$, and $h \in_u \hh$ with probability at least $1-\epsilon$, then for all $j\in [m]$, $|h^{-1}(j)\cap I| \leq |I|/\ell$.
\end{definition}

The above definition quantifies the intuition that when a sufficiently large (so that standard tail bounds apply) collection of items $I \subseteq [n]$ is hashed into $m$ bins, the max-load is not much more than the average load of $|I|/m$. It will be important for us to be able to construct such families explicitly.

\begin{lem}\label{SpreadingHashLem}
For all $\epsilon \geq 0$, there exists an explicit hash family $\hh = \{h:[n] \to [m]\}$ where $m = O(\log^{5}(1/\epsilon))$ which is $((\log^{5} (1/\epsilon)), \log(1/\epsilon),\epsilon)$-spreading and $h \in_u \hh$ can be sampled with $O(\log(n/\epsilon))$ bits.
\end{lem}
\begin{proof}
Let $m=\Theta(\log^5(1/\epsilon))$ and let $\hh = \{h:[n]\rightarrow[m]\}$ be a $\delta$-biased family for $\delta=\exp(-C(\log(1/\epsilon)))$ for $C$ a sufficiently large constant. We argue that $\hh$ satisfies the conditions of the lemma by standard moment bounds.

Let $p=2\log(1/\epsilon)/\log\log(1/\epsilon))$. Let $|I|>\log^5(1/\epsilon)$ and let $v\in\{0,1\}^n$ be the indicator vector of the set $I$. Note that if some $h$ has $|h^{-1}(j)\cap I| > |I|/\log(1/\epsilon)$ for some $j$, then $h(v) \geq |I|^2/\log^2(1/\epsilon)$ (recall the definition of $h(v)$ from Equation \ref{eq:hv}). Therefore, by Lemma \ref{hashmomentsLem} and Markov's inequality, the probability that this happens is at most
\begin{align*}
\frac{\E[h(v)^p]\log^{2p}(1/\epsilon)}{|I|^{2p}} & \leq O\left(\frac{p^2\log^2(1/\epsilon)}{m}+p^2\log^2(1/\epsilon)|I|^{-1} \right)^p + m^p \log^{2p}(1/\epsilon)\delta\\
& = O(\log(1/\epsilon))^{-p} + O(\log(1/\epsilon))^{5p} \delta\\
& \leq \epsilon.
\end{align*}
\end{proof}

\subsubsection{The $\prg$}
We begin with a simpler version of our generator which has the desired pseudorandomness property but has too large a seed. We will then improve the seed-length using PRGs for small-space machines.

Let $\hh = \{h:[n] \to [m]\}$ be a $(k,C\log(1/\epsilon),\epsilon)$-spreading family for parameters $k,C,\epsilon$ to be chosen later. Let $\gcs:\zo^r \to \dpm^n$ be a generator as in Theorem \ref{simpleGeneratorThm} with error $1/4$. Now, define the generator $\glarge:\hh \times \left(\zo^r\right)^m \to \dpm^n$ as follows: for $i \in [n]$,
\begin{equation}
  \label{GenLargeEqn}
  \glarge(h,z_1,z_2,\ldots,z_m)_i = \gcs(z_{h(i)})_i.
\end{equation}

We claim that the above generator fools tests of the form $\phi_{v,\alpha}(\;)$ for $\|v\|_0\geq k$ and $\alpha \gg \sqrt{m}/\|v\|_2$.

\begin{lem}\label{LargeGen1Lem}
Let $C$ be a sufficiently large constant. Let $\hh = \{h:[n] \to [m]\}$ for some $m\geq \log(1/\epsilon)$ be a $(k,\ell,\epsilon/4)$-spreading family with $\ell=C\log(1/\epsilon)$. Let $\gcs$ be a generator as in Theorem \ref{simpleGeneratorThm} with error $1/4$. Let $Y \in \dpm^n$ be the output of the generator $\glarge$ as defined in Equation \eqref{GenLargeEqn} on a uniformly random seed and $X \in_u \dpm^n$. Then, for all $v \in \zpm^n$ with $\|v\|_0 \geq k$, and $ C\sqrt{m}/\|v\|_2 \leq \alpha \leq 1/4$,
$$\left|\E[\phi_{v,\alpha}(Y)] - \E[\phi_{v,\alpha}(X)]\right| \leq \epsilon.$$
\end{lem}
\begin{proof}
Fix the test vector $v \in \zpm^n$. Let $I = Supp(v)$ and let $|I| = K \geq k$. Let $Y = \glarge(h,z_1,z_2,\ldots,z_m)$ and for $j \in [m]$, let $Y^j = \gcs(z_j)$ and let $X^j \in_u \dpm^n$ be independent uniformly random strings. Suppose that the hash function $h \in_u \hh$ is such that the condition of $(k,\ell,\epsilon/4)$-spreading holds for $I$. This assumption only incurs an additive $\epsilon/2$ in the error.

First note that,
$$\E[\phi_{v,\alpha}(X)] = (\cos 2\pi \alpha)^K \leq \exp(-\Omega(\alpha^2 K)) \leq \exp(-Cm)\leq \epsilon/4.$$
Thus, we need only show that
$$\E[\phi_{v,\alpha}(Y)]\leq \epsilon/4.$$

Now, for $j \in [m]$ let $v^j = v_{h^{-1}(j)}$ and $K_j = |I \cap h^{-1}(j)|$. Observe that by definition, $\dtv(v^j\cdot Y^j, v^j \cdot X^j) \leq 1/4$. Therefore, $$\left|\E[\phi_{v^j,\alpha}(Y^j)] - \E[\phi_{v^j,\alpha}(X^j)]\right| \leq 1/2.$$
Further,
$$\E[\phi_{v^j,\alpha}(X^j)] = (\cos 2\pi\alpha)^{K_j} = \exp(-\Omega(\alpha^2 K_j)).$$
Combining the above two equations, we get
$$\left|\E[\phi_{v,\alpha}(Y)]\right| = \left|\prod_{i=1}^m \E[\phi_{v^j,\alpha}(Y^j)]\right| \leq \prod_{i=1}^m \min\left(\left(\frac{1}{2} + \exp(-\Omega(\alpha^2 K_j))\right),1\right).$$

Now, because $h$ has the well-spreading property, $K_j = |h^{-1}(j) \cap I| \leq |I| /\ell$ for all $j \in [m]$. On the other hand, $\sum_j K_j = K$. Since the sum of the $K_j$ which are at most $K/(2m)$ totals at most $K/2$ and since none of the other $K_j$ are too large, there must be at least $\ell/2$ values of $j$ so that $K_j \geq K/(2m)$. For these $j$ we have that
$$
\frac{1}{2} + \exp(-\Omega(\alpha^2 K_j)) \leq \frac{1}{2} + \exp(-\Omega((C^2m/K) (K/2m))) = \frac{1}{2} + \exp(-\Omega(C)) \leq \frac{3}{4}
$$
for $C$ sufficiently large. Thus, for $C$ sufficiently large
$$
\left|\E[\phi_{v,\alpha}(Y)]\right| \leq \left( \frac{3}{4}\right)^{\ell/2} \leq \epsilon/4.
$$
This completes the proof.

\end{proof}

We are now ready to prove Proposition \ref{genLargeProp}.
\begin{proof}[Proof of Proposition \ref{genLargeProp}]
Let $C$ be a sufficiently large constant, $m=C\log^5(1/\epsilon)$, let $\hh=\{h:[n]\rightarrow[m]\}$ be a $(k,\ell,\epsilon/4)$-spreading family with $k\leq C\log^5(1/\epsilon)$, and $\ell=C\log(1/\epsilon)$ as given in Lemma \ref{SpreadingHashLem}. Note that if $1/4 \geq \alpha \geq \log^3(1/\epsilon)/\|v\|_2$ for some $\alpha$, it must be the case that $\|v\|_0\geq \log^6(1/\epsilon)\geq k$. Therefore, Lemma \ref{LargeGen1Lem} provides us with a generator, $\glarge$, so that for any such $\alpha$ that if $Y$ is an output of $\glarge$ and $X$ a uniform random element of $\{\pm 1\}^n$ and if $\nm{v}_0\geq k$, then
$$
\left| \E[\phi_{v,\alpha}(Y)]-\E[\phi_{v,\alpha}(X)]\right| \leq \epsilon/2.
$$
Unfortunately, the seed-length of $\glarge$ is $\log(\hh) + O(\log n)\cdot m$. We improve this using the PRGs for ROBPs of Theorem \ref{th:inwprg}. It is easy to see that for a fixed hash function $h$ and test vector $v$, the computation of $\dotp{v}{\glarge(h,z_1,\ldots,z_m)}$ can be done by a $(S,D,m)$-ROBP where $S = O(\log n)$ and $D = O(\log n)$. Thus, we can further derandomize the choice of $z_1,\ldots,z_m$ using the PRG from Theorem \ref{th:inwprg}. Formally, let $G^{INW}:\zo^r \to \left(\zo^D\right)^m$ be a generator fooling $(S,D,m)$-ROBPs as in Theorem \ref{th:inwprg} with error $\epsilon/4$ and define
$$G^f(h,z) = \glarge(h,G^{INW}(z)).$$

Then, from the above arguments it follows that $G^f$ fools $\phi_{v,\alpha}$ with error at most $\epsilon$ and has seed-length $O(\log (n/\epsilon) \cdot(\log\log(n/\epsilon)))$ proving the claim.
\end{proof}

\subsection{Generator for small $\alpha$}
We next argue that the generator defined in Equation \ref{LargeAlphaEqn} fools Fourier coefficients $\phi_{v,\alpha}$ for sufficiently small $\alpha$.
\ignore{\subsubsection{The $\prg$}\label{sec:smallalphagenerator}
Let $n,\delta>0$. Let $C$ be a sufficiently large constant. We define a generator as follows. Let $n=n_1>n_2>\ldots > n_t$ so that $n_{i+1}=n_i^{1/2}+O(1)$ and $\log^{2C}(n/\delta)\geq n_t \geq \log^C(n/\delta)$. Note that this implies that $t=O(\log\log(n))$. For $1\leq i < t$ let $h_i:[n_i]\to [n_{i+1}]$ be a random hash function chosen from a $\frac{C\log(n/\delta)}{\log(n_i)}$-wise independent family. Let $Z_i$ be a random element of $\dpm^{n_i}$ chosen from a distribution that is both $(\delta/n)^C$-biased and and $\frac{C\log(n/\delta)}{\log(n_i)}$-wise independent. Let $Z$ be a random variable in $\dpm^{n_t}$ be chosen to fool weight at most $n$ halfspaces to variational distance $\delta/n$ as described in Theorem \ref{simpleGeneratorThm}. Finally, we define our random variable $Y\in\dpm^n$ to be
\begin{equation}\label{LargeAlphaEqn}
Y = Z A(h_{t-1})D(Z_{t-1})A(h_{t-2})D(Z_{t-2})\cdots A(h_1)D(Z_1).
\end{equation}

Informally, this generator begins with the generator $X_1$, then uses $h_1$ to divide the coordinates into $n_2$ bins and then for each bin multiplies the elements in this bin by a random sign, these $n_2$ signs being chosen recursively by a similar generator, until at the final level they are picked using the generator from Theorem \ref{simpleGeneratorThm} instead.

It is easy to see from Theorem \ref{simpleGeneratorThm} and Fact \ref{HashFamilyFact} that the random variable $Y$ can be produced from a random seed of length $s=O(\log(n/\delta)\log\log(n/\delta))$. We also claim that it fools $\phi_{v,\alpha}$ for $|\alpha| \leq \log^3(1/\delta)/\|v\|_2$.}
The main claim of this section is the following.
\begin{prop}\label{SmallAlphaProp}
Let $v\in\zpm^n$ and $\alpha\in\R$ with $|\alpha| \leq \log^3(1/\delta)/\|v\|_2$. Let $C$ be a sufficiently large constant and let $\delta>0$. Let $Y$ be as defined by Equation \eqref{LargeAlphaEqn} and let $X\in_u \dpm^n$. Then
$$
\left|\E[\phi_{v,\alpha}(Y)]-\E[\phi_{v,\alpha}(X)] \right| \leq \delta.
$$
\end{prop}
As described in the overview, we will prove the claim by a hybrid argument. For ease of notation, we repeat some notation from the overview section.
For $1\leq i <t$, letting $X_i$ be a uniform random element of $\dpm^{n_i}$ we define
\begin{equation}\label{YiEqn}
Y_i := X_iA(h_{i-1})D(Z_{i-1})\cdots A(h_1)D(Z_1)
\end{equation}
(note that $Y_1=X_1$) and let $Y_t=Y$. Our Proposition will follow from the following Lemma.
\begin{lem}\label{SmallAlphaLem}
With $Y_i$ defined as above for $C$ sufficiently large and $v\in\zpm^n$ and $\alpha\in\R$ with $|\alpha| \leq \log^3(1/\delta)/\|v\|_2$, then for $t> i \geq1$
$$
\left| \E[\phi_{v,\alpha}(Y_{i+1})] -  \E[\phi_{v,\alpha}(Y_{i})]\right| \leq \delta/n.
$$
\end{lem}

The proof of the Lemma \ref{SmallAlphaLem} will be further split into two main cases based upon whether or not the vector $v$ is sparse relative to $n_i$. Intuitively, the case of sparse $v$ is easier as hashing takes care of most issues here.
\subsubsection{Analysis for sparse vectors}
We begin with the case where $v$ is sparse.
\begin{lem}\label{sparseVLem}
With $Y_i,C,n,v,\alpha,\delta$ as in Lemma \ref{SmallAlphaLem} with $i<t$, if $\|v\|_0^3 < n_{i+1}$ then
$$
\left| \E[\phi_{v,\alpha}(Y_{i+1})] -  \E[\phi_{v,\alpha}(Y_{i})]\right| \leq \delta/n.
$$
\end{lem}
\begin{proof}
We claim that this holds even after fixing the values of $h_j,Z_j$ for all $j<i$. In particular, if we let
$$
w =vD(X_1)A(h_1)^T\cdots D(h_{i-1})A(h_{i-1})^T
$$
then we need to show that
$$
\left|\E[\phi_{w,\alpha}(X_{i+1}A(h_i)D(Z_i))] - \E[\phi_{w,\alpha}(X_{i})] \right| \leq \delta/n.
$$
We will show the stronger claim that
$$
\dtv(w\cdot X_{i+1}A(h_i)D(Z_i),w\cdot X_i) \leq \delta/(2n).
$$
Intuitively, this will hold because $v$ (and hence $w$) is sparse. This means that with high probability $h_i$ will cause few collisions within the support of $w$. If this is the case, then $Z_i$ will nearly randomize the relative signs of elements mapped to the same bin and $X_i$ will randomize the signs between bins. To show that we have few collisions, we will need the following lemma:
\begin{lem}\label{collisionLem}
Let $n$ and $m$ be positive integers, $\epsilon>0$ and $C$ a sufficiently large constant. Let $\hh=\{h:[n]\to [m]\}$ be a $k$-wise independent family of hash functions for $k=\frac{C\log(m/\epsilon)}{\log(m)}$. Let $I\subset [n]$ be such that $|I|^3 \leq m$. Then for $h\in_u \hh$, with probability at least $1-\epsilon$ we have that
$$
|I|-|h(I)|\leq k.
$$
\end{lem}
\begin{proof}
Note that if $|I|-|h(I)|>k$ then at least $k$ elements of $I$ were sent to the same location as some other element of $I$. This implies that there must be at least $k/3$ disjoint pairs of elements $x_i,y_i\in I$ so that $h(x_i)=h(y_i)$ (for each element $j\in[m]$ so that $|h^{-1}(j)|=\ell>1$ we can find at least $\ell/3$ pairs). Thus, it suffices to show that the expected number of collections of distinct elements $x_1,y_1,x_2,y_2,\ldots,x_{k/3},y_{k/3}\in I$ so that $h(x_i)=h(y_i)$ for each $i$ is less than $\epsilon$. On the other hand, the number of sequences $x_i,y_i\in I$ is at most $|I|^{2k/3}$ and the probability that any given sequence has the desired property is $m^{-k/3}$ by $k$-wise independence of $h$. Thus the expected number of such sets of pairs is at most
$$
|I|^{2k/3}m^{-k/3} \leq m^{2k/9}m^{-k/3} = m^{-k/9} \leq \epsilon.
$$
This completes the proof.
\end{proof}

Applying this lemma to $I=\supp(w)$, we find that except with probability $\delta/(4n)$ we have that at most $\log(n/\delta)$ elements of $I$ collide with any other element of $I$ under $h_i$. Let $J$ be the set of such coordinates. It is clear that the distribution of $w\cdot (X_{i+1}A(h_i)D(Z_i))$ as we vary $X_{i+1}$ depends only on $h_i$ and the signs of the $Z_i$ on the coordinates of $J$. On the other hand, it is easy to see that the restriction of $Z_i$ to these coordinates is within $2^{|J|}(\delta/n)^C < \delta/(4n)$ of uniform. Thus,
$$
\delta/(2n) \geq \dtv(w\cdot (X_{i+1}A(h_i)D(Z_i)),w\cdot (X_{i+1}A(h_i)D(X_i))) = \dtv(w\cdot X_{i+1}A(h_i)D(Z_i),w\cdot X_i).
$$
Where the equality above is because $(X_{i+1}A(h_i)D(X_i))$ and $X_i$ are identically distributed.
This completes the proof.
\end{proof}

\subsubsection{Analysis for dense vectors}

For relatively dense vectors $v$, we will need a different, more analytic approach. The following crucial lemma analyzes the effect of a single dimension reduction step and bounds the error in terms of the norms of the test vector $v$. We will then apply the lemma iteratively.
\begin{lem}\label{denseVLem}
Let $\delta > 0$, $n,m \geq 1$, and $p\geq 2$ and even integer. Let $\calD$ be a $2p$-wise independent distribution over $\dpm^n$ and $\hh = \{h:[n] \to [m]\}$ be a $2p$-wise independent hash family. Then, for all $v \in \R^n$, $X \in_u \dpm^n$,$Y \sim \calD$, $h\in_u \hh$ and $Z \in_u \dpm^m$,
\begin{multline}
|\E[\phi_{v,\alpha}(Z \cdot A(h) \cdot D(Y))] - \E[\phi_{v,\alpha}(Z \cdot A(h) \cdot D(X))]| < O(p)^{2p}\left(\frac{\alpha^4\|v\|_2^4}{m} +\alpha^4\|v\|_4^4 \right)^{p/8}.
\end{multline}
\end{lem}

To prove the lemma we shall exploit the independence of the $z_i$'s in Equation \eqref{gensmallEqn} to reduce the problem to that of analyzing a product of cosines as in the following lemma. The lemma gives a low-degree (multivariate) polynomial approximation for a product of cosines.
\begin{lem}\label{cosApproxLem}
For all $\alpha \in (0,1/4)$ and even integer $p$, there exists a polynomial $P: \R^m \to \R$ of degree at most $p$ such that for all $S_1,\ldots,S_m, T \in \R$,
 \begin{align}
   \label{CosApproxEqn}
   & \prod_{j=1}^m \cos(2\pi\alpha S_j) = \exp(-2\pi^2\alpha^2T) \cdot \left(\sum_{t=0}^{p/2-1} \frac{\left(-2\pi^2\alpha^2\left(\sum_{i=1}^m S_i^2-T\right)\right)^t}{t!} P(S_i)\right)\\ & + O(1)^p\left(\left( \alpha^2\left(\sum_{i=1}^m S_i^2-T\right)\right)^{p/2}+\left( \alpha^2\left(\sum_{i=1}^m S_i^2-T\right)\right)^{p}+\left(\sum_{i=1}^m (\alpha S_i)^{4}\right)^{p/8}+\left(\sum_{i=1}^m (\alpha S_i)^{4}\right)^{p/2}\right).\nonumber
\end{align}
\end{lem}

\begin{proof}[Proof of Lemma \ref{denseVLem}]
Let $Y^s = Z A(h) D(Y)$. We first fix a hash function $h \in \hh$ and then bound the error as a function of the hash function. We then average the error bound for a uniformly random hash function from $\hh$ using Lemma \ref{hashmomentsLem}.

For $j \in [m]$, let random variable $S_j = \sum_{i: h(i) = j} v_i Y_i$. Note that $\iprod{v}{Y^s} = \sum_{j=1}^m z_j S_j$. Therefore, as $z \in_u \dpm^m$,
$$\E_z[\phi_{v,\alpha}(Y^s)] = \prod_{j=1}^m \cos(2\pi\alpha S_j).$$
Let $T = \sum_j \E[S_j^2] = \|v\|_2^2$. Let $Q(X) \equiv R(S_1,\ldots,S_m)$ denote the degree $2p$ polynomial corresponding in the first term of Equation \eqref{CosApproxEqn}, and let $E(X)$ be the \emph{error} term corresponding to the second term. Then, from the above calculations,
$$\E_z[\phi_{v,\alpha}(X^s)] = Q(X) + E(X).$$

Observe that
\begin{equation}
  \label{gensmallEqn}
Q_2(X) := \alpha^2 (\sum_{j=1}^m S_j^2 - T) = \alpha^2 \sum_{j=1}^m \sum_{i \neq i' \in h^{-1}(j)} v_i v_{i'} X_i X_{i'},
\end{equation}
is a degree two polynomial in $X$ with, $\|Q_2\|_2^2 \leq \alpha^4 h(v)$ (recall Equation \eqref{eq:hv}).

By hypercontractivity - Lemma \ref{hcLem}, for all even $r\leq p$,
$$\E[Q_2(X_1,\ldots,X_n)^r] \leq O(r)^{r} \bkets{\alpha^4 h(v)}^{r/2}.$$

A similar calculation for the polynomial $Q_4(X_1,\ldots,X_n) := \alpha^4 (\sum_j S_j^4)$ shows that for all even $r\leq p/2$,
$$\E[Q_4(X_1,\ldots,X_n)^r] \leq O(r)^{2r} \bkets{\alpha^4 h(v)}^{r}.$$

By $2p$-wise independence, the above bounds also hold for $\E[Q_2(Y)^r],\E[Q_4(Y)^r]$.

Now, let $X^s = Z A(h) D(X)$, where $X \in_u \dpm^n$. Then, clearly $X^s \in_u \dpm^n$. Combining the above expressions and noting that they also work for $X \in_u \dpm^n$, we get
\begin{align*}
\left|\E_X  \E_z[\phi_{v,\alpha}(X^s)] - \E_Y \E_z[\phi_{v,\alpha}(Y^s)]\right| &\leq \left|\E\sbkets{Q(X) - Q(Y)}\right| + \E[|E(X)|] + \E[|E(Y)|] \\
&\leq 0 + O(p)^p \bkets{\alpha^4 h(v)}^{p/4} + O(p)^p \bkets{\alpha^4 h(v)}^{p/8}\\
& \leq O(p)^p \bkets{\alpha^4 h(v)}^{p/8}.
\end{align*}

By taking expectation with respect to $h \in_u \hh$ and applying Lemma \ref{hashmomentsLem}, we get
\begin{align*}
  \left|\E[\phi_{v,\alpha}(X^s)] - \E[\phi_{v,\alpha}(Y^s)]\right| &\leq O(p)^{2p}\left(\frac{\alpha^4\|v\|_2^4}{m} +\alpha^4\|v\|_4^4 \right)^{p/8},
\end{align*}
proving the lemma.
\end{proof}

We defer the proof of Lemma \ref{cosApproxLem} to Section \ref{cosApprox} and continue with the analysis of our generator. We do so by applying Lemma \ref{denseVLem} iteratively to the vectors
$$
v_i := vD(Z_1)A(h_1)^T\cdots D(Z_{i-1})A(h_{i-1})^T.
$$
In order for it to be useful, we need to have good bounds on the low order moments of the $v_i$. We deal with these issues in the next section.
\subsubsection{Controlling moments}
In particular we will need the following Lemma:
\begin{lem}\label{momentControlLem}
Let $v\in\zpm^n$ with $\|v\|_0\geq \log^{C/4}(n/\delta)$. Let $Z_i,h_i,v_i$ be defined as above. For any $1\leq i \leq t$ we have with probability at least $1-\delta/(4n)$ that
$$
\|v_i\|_2 \leq 2^i \|v\|_2 \ \ \ \ \ \mathrm{and} \ \ \ \ \ \|v_i\|_4 \leq \frac{\|v\|_2}{\min(\|v\|_2^{1/3},n_i^{1/20})}.
$$
\end{lem}

In order to prove this we will first need some controls over how the procedure used to obtain $v_{i+1}$ from $v_i$ affects these norms. In particular, we show:
\begin{lem}\label{NormControlLem}
Let $p\geq 2$ be an even integer. Let $\hh = \{h: [n] \to [m]\}$ be a $4p$-wise independent hash family and $\calD$ be a $4p$-wise independent distribution over $\dpm^n$. Then, for $h \in_u \hh$, $x \sim \calD$ and a vector $v \in \R^n$,
$$\E\left[\left(\|v\|_2^2-\nm{v D(x)A(h)^T}_2^{2}\right)^p\right] \leq O(p)^{2p} \left(\frac{\|v\|_2^4}{m}\right)^{p/2} + O(p)^{2p} \|v\|_4^{2p}.$$
Similarly,
$$\E[\nm{v D(x)A(h)^T}_4^{4p}] \leq O(p)^{4p}\left(\frac{\|v\|_2^4}{m} \right)^p + O(p)^{4p} \|v\|_4^{4p}.$$
\end{lem}
\begin{proof}
Note that in either case the independence is sufficient that the expectations would be the same if $x$ and $h$ were chosen uniformly at random from $\dpm^n$ and $[m]^{[n]}$, respectively.

Applying Lemma \ref{hcLem} to the polynomial $P_h(x) = \|v D(x)A(h)^T\|_2^2-\|v\|_2^2$, we find that for fixed $h$
$$
\E_x[P_h(x)^p] \leq O(p)^p h(v)^{p/2}.
$$
Averaging over $h$ and applying Lemma \ref{hashmomentsLem} yields the first line.

Applying Lemma \ref{hcLem} to the polynomial $Q_h(x) = \|v D(x)A(h)^T\|_4^4$, we find that for fixed $h$,
$$
\E_x[\nm{v D(x)A(h)^T}_4^{4p}] \leq O(p)^{2p}h(v)^p.
$$
Taking an expectation over $h$ and applying Lemma \ref{hashmomentsLem}, we get that
$$
\E[\nm{v D(x)A(h)^T}_4^{4p}] \leq O(p)^{4p}\left(\frac{\|v\|_2^4}{m} \right)^p + O(p)^{4p} \|v\|_4^{4p}.
$$
This completes the proof.
\end{proof}

We are now prepared to prove Lemma \ref{momentControlLem}.

\begin{proof}[Proof of Lemma \ref{momentControlLem}]
We proceed by induction on $i$ proving that the desired inequalities hold with probability at least $1-i(\delta/n)^2$. As a base case we consider $i$ so that $\|v\|_0^3 \leq n_i$. In this case, by repeated application of Lemma \ref{collisionLem}, we find that with at least the desired probability that $\|v\|_0-\|v_i\|_0 \leq i\log(n/\delta).$ This implies that other than its zero coefficients, $v_i$ has $\|v\|_0-2i\log(n/\delta)$ coefficients of norm 1, and at most $i\log(n/\delta)$ other coefficients each of norm at most $i\log(n/\delta)$. This means that
$$
\|v_i\|_2^2 = \|v\|_2^2 + O(i^3 \log^3(n/\delta)), \ \ \ \ \mathrm{and} \ \ \ \ \|v_i\|_4^4 = \|v\|_2^2 + O(i^5 \log^5(n/\delta)).
$$
Our bounds follow immediately.

Otherwise, for $\|v_0\|^6 > n_i$, we proceed by induction on $i$. As a base case, note that the desired inequalities hold for $i=1$ as $v_1=v$, and $\|v\|_4=\sqrt{\|v\|_2}.$ We claim that if $\|v_i\|$ satisfies the desired inequalities, then $v_{i+1}$ also does with probability at least $1-(\delta/n)^2$. Note that $v_{i+1} = v_i D(Z_i)A(h_i)^T$. Note also that $Z_i$ and $h_i$ are $k$-wise independent for $k=C\log(n/\delta)/\log(n_i)$. Applying Lemma \ref{NormControlLem} with $p=\floor{k/4}$, we find that
$$
\E\left[\left(\|v_i\|_2^2-\|v_{i+1}\|_2^2\right)^{2p}\right] \leq O(p)^{4p} \left(\frac{\|v_i\|_2^4}{n_{i+1}}\right)^{p} + O(p)^{4p} \|v_i\|_4^{4p},
$$
and
$$
\E[\|v_{i+1}\|_4^{4p}] \leq O(p)^{4p}\left(\frac{\|v_i\|_2^4}{n_{i+1}} \right)^p + O(p)^{4p} \|v_i\|_4^{4p}.
$$

Applying the Markov bound to the first of these equations we find that the probability that $\|v_{i+1}\|_2^2 \geq \|v_i\|_2^2+4^i\|v\|_2^2$ is at most
\begin{align*}
\left(\frac{O(p^4)}{n_{i+1}}\right)^{p} + O\left( \frac{p\|v_i\|_4}{\|v\|_2}\right)^{4p} & \leq n_{i+1}^{-p/2} + O(pn_{i+1}^{-1/10})^{4p}\\
\leq n_{i+1}^{-p/2} + n_{i+1}^{-p/11}\\
\leq (\delta/n)^2/2.
\end{align*}
Where the first inequality above is by the inductive hypothesis. This implies that $\|v_{i+1}\|_2 \leq 2^{i+1}\|v\|_2$ with the desired probability.

Applying the Markov bound to the latter of these equations we find that the probability that $\|v_{i+1}\|_4 > \|v\|_2 /n_{i+1}^{1/20}$ is at most
\begin{align*}
O\left(\frac{p^4\|v_i\|_2^4}{\|v\|_2^4 n_{i+1}^{4/5}} + \frac{p^4\|v_i\|_4^4n_{i+1}^{1/5}}{\|v_2\|_2^4} \right)^p & \leq O\left( n_i^{-1/2} + \frac{n_{i+1}^{1/5}}{n_{i}^{1/5}} \right)^p \leq O(n_i^{-1/10})^p \leq (\delta/n)^2/2.
\end{align*}
Where above we use that $$\|v_i\|_4 \leq \frac{\|v\|_2}{\min(\|v\|_2^{1/3},n_i^{1/20})}=\frac{\|v\|_2}{n_i^{1/20}}.$$
Thus, with the desired probability $\|v_{i+1}\|_4 \leq \|v\|_2 /n_i^{1/20}$. This completes the inductive step, and finishes the proof.
\end{proof}

\subsubsection{Combined analysis}
We are now ready to prove Lemma \ref{SmallAlphaLem}.
\begin{proof}[Proof of Lemma \ref{SmallAlphaLem}]
First note that if $i=t-1$, the lemma follows immediately from the pseudorandomness properties of $Z$. We thus consider only $i<t-1$.

We note that $v\cdot Y_i = v_i\cdot X_i$ and $v\cdot Y_{i+1} = v_i \cdot X_{i+1}A(h_i)D(Z_i)$. If $\|v\|_0^3 \leq n_{i+1}$, we are done by Lemma \ref{sparseVLem}. Otherwise, assume that $\|v\|_0^3 > n_{i+1}$. By Lemma \ref{momentControlLem} we have that except for an event of probability $\delta/(4n)$ we have that
$$
\|v_i\|_2 \leq \log(n)\|v\|_2 \ \ \ \mathrm{and} \ \ \ \|v_i\|_4 \leq \|v\|_2 n_i^{-1/20}.
$$
By ignoring the possibility that these are violated, we introduce an error of at most $\delta/(2n)$, thus it suffices to only consider the case where the choice of $h_1,Z_1,\ldots,h_{i-1},Z_{i-1}$ are such that the above holds. We now need to bound
$$
\left|\E[\phi_{v_i,\alpha}(X_i)] - \E[\phi_{v_i,\alpha}(X_{i+1}A(h_i)D(Z_i))] \right|.
$$
Since $X_i$ has the same distribution as $X_{i+1}A(h_i)D(X_i)$, we may apply Lemma \ref{denseVLem} that for $p=\Omega\left(\frac{C\log(n/\delta)}{\log(n_i)}\right)$ that the above is bounded by
\begin{align*}
O(p)^{2p}\left(\frac{\alpha^4 \|v_i\|_2^4}{n_{i+1}}+ \alpha^4\|v_i\|_4^4\right)^{p/8} & \leq O(p)^{2p} \left(2^{4i}\log^{12}(1/\delta)n_{i}^{-1/2}+\log^{12}(1/\delta)n_i^{-1/5} \right)^{p/8}\\
& \leq \left(\log^{52}(n/\delta)n_i^{-1/5} \right)^{p/8}\\
& \leq n_i^{-p/50}\\
& \leq \delta/(2n).
\end{align*}
This completes the proof.
\end{proof}

Proposition \ref{SmallAlphaProp} now follows immediately after noting that
$$
\left|\E[\phi_{v,\alpha}(X)]- \E[\phi_{v,\alpha}(Y)]\right| \leq \sum_{i=1}^{t-1}\left|\E[\phi_{v,\alpha}(Y_{i})]- \E[\phi_{v,\alpha}(Y_{i+1})]\right|.
$$

\subsubsection{Approximating a product of cosines}\label{cosApprox}    
Here we prove Lemma \ref{cosApproxLem}.
\begin{proof}[Proof of Lemma \ref{cosApproxLem}]
Note that so long as $\alpha S_i < 1/10$ for all $i$ that by Taylor expansion we have that
\begin{align*}
& \phantom{=}\prod_{i=1}^m \cos(2\pi \alpha S_i)\\ & = \exp\left(-2\pi^2\alpha^2\sum_{i=1}^m S_i^2+\sum_{j=2}^{p/2-1}\left(c_j\sum_{i=1}^m (\alpha S_i)^{2j}\right) + \sum_{i=1}^m O(\alpha S_i)^{p} \right)\\
& = \exp(-2\pi^2\alpha^2T)\exp\left(-2\pi^2\alpha^2\left(\sum_{i=1}^m S_i^2-T\right)+\sum_{j=2}^{p/2-1}\left(c_j\sum_{i=1}^m (\alpha S_i)^{2j}\right) + \sum_{i=1}^m O(\alpha S_i)^{p} \right)
\end{align*}
where the $c_j$ are constants obtained from the Taylor expansion of $\log(\cos(z))$. Furthermore, since $\log(\cos(z))$ is analytic in a disk around $z=0$, we have that $c_j = O(1)^j$. Note by conditioning on whether or not $ \sum_{i=1}^m O(\alpha S_i)^{p}$ is more than $1$, we find that the above is equal to
$$
\exp(-2\pi^2\alpha^2T)\exp\left(-2\pi^2\alpha^2\left(\sum_{i=1}^m S_i^2-T\right)+\sum_{j=2}^{p/2-1}\left(c_j\sum_{i=1}^m (\alpha S_i)^{2j}\right) \right)\left( 1 +  \sum_{i=1}^m O(\alpha S_i)^{p}\right) +  \sum_{i=1}^m O(\alpha S_i)^{p}.
$$
For each $j$, let $p_j$ be the ceiling of $p/(2j)$. Note that $p\leq 2j\cdot p_j \leq 2p$. Under the additional assumption that $\alpha^2\left(\sum_{i=1}^m S_i^2-T\right),\sum_{i=1}^m (\alpha S_i)^{4}<a$, for some sufficiently small constant $a$ we have that the above is equal to
\begin{align*}
\exp(-2\pi^2\alpha^2T)&\cdot \left(\sum_{t=0}^{p_2-1} \frac{\left(-2\pi^2\alpha^2\left(\sum_{i=1}^m S_i^2-T\right)\right)^t}{t!} \right)\cdot \left(1+O\left( \alpha^2\left(\sum_{i=1}^m S_i^2-T\right)\right)^{p_2}\right)\\
& \cdot \prod_{j=2}^{p/2-1}\left(\sum_{t=0}^{p_j-1}\frac{\left(c_j\sum_{i=1}^m (\alpha S_i)^{2j}\right)^t}{t!} \right)\cdot \left(1+O(1)^p\left(\sum_{i=1}^m (\alpha S_i)^{2j}\right)^{p_j}\right)\\
& \cdot \left( 1 +  \sum_{i=1}^m O(\alpha S_i)^{p}\right) +  \sum_{i=1}^m O(\alpha S_i)^{p}\\
= \exp(-2\pi^2\alpha^2T)&\cdot \left(\sum_{t=0}^{p_2-1} \frac{\left(-2\pi^2\alpha^2\left(\sum_{i=1}^m S_i^2-T\right)\right)^t}{t!} \right)\cdot \prod_{j=2}^{p/2-1}\left(\sum_{t=0}^{p_j-1}\frac{\left(c_j\sum_{i=1}^m (\alpha S_i)^{2j}\right)^t}{t!} \right)\\
& + O(1)^p\left(\left( \alpha^2\left(\sum_{i=1}^m S_i^2-T\right)\right)^{p_2}+\left(\sum_{i=1}^m (\alpha S_i)^{4}\right)^{p/4}+ \sum_{i=1}^m (\alpha S_i)^{p}\right)
\end{align*}
Next consider the above term
$$
\prod_{j=2}^{p/2-1}\left(\sum_{t=0}^{p_j-1}\frac{\left(c_j\sum_{i=1}^m (\alpha S_i)^{2j}\right)^t}{t!} \right).
$$
Let it equal $P(S_i)+E(S_i)$ where $P$ is the polynomial consisting of all the terms of total degree at most $p$. We note that for any $j$ that $c_j\sum_{i=1}^m (\alpha S_i)^{2j}$ is at most $$O(a)^{j/4} \left(\sum_{i=1}^m (\alpha S_i)^{4}\right)^{ j/4 }.$$ Therefore, $|E(S_i)|$ is at most $\left(\sum_{i=1}^m (\alpha S_i)^{4}\right)^{p/8}$ times the sum of the degree more than $p$ coefficients in the Taylor expansion of
$$
\exp\left( \frac{1}{1-O(a^{1/8}z)} \right).
$$
For $a$ sufficiently small, the above has radius of convergence more than $1$, and thus the sum of the degree more than $p$ terms is bounded.
Thus, $E(S_i)$ is
$$
O\left(\sum_{i=1}^m (\alpha S_i)^{4} \right)^{p/8}.
$$
Therefore, assuming that $\alpha S_i < 1/10$ for all $i$, and $\alpha^2\left(\sum_{i=1}^m S_i^2-T\right),\sum_{i=1}^m (\alpha S_i)^{4}<a$, then
$$\prod_{i=1}^m \cos(2\pi \alpha S_i)$$ equals
\begin{align*} \exp(-2\pi^2\alpha^2T)& \cdot \left(\sum_{t=0}^{p_2-1} \frac{\left(-2\pi^2\alpha^2\left(\sum_{i=1}^m S_i^2-T\right)\right)^t}{t!} \right)P(S_i)\\ & + O(1)^p\left(\left( \alpha^2\left(\sum_{i=1}^m S_i^2-T\right)\right)^{p_2}+\left(\sum_{i=1}^m (\alpha S_i)^{4}\right)^{p/8-1}+ \sum_{i=1}^m (\alpha S_i)^{p}\right).
\end{align*}
On the other hand, if the stated assumptions fail, the main term above is bounded by a polynomial in $\alpha^2\left(\sum_{i=1}^m S_i^2-T\right)$ and $\sum_{i=1}^m (\alpha S_i)^{4}$ with total degree at most $2p$ and sum of coefficients $O(1)^p$. Therefore, under no additional assumptions we have that
$$\prod_{i=1}^m \cos(2\pi \alpha S_i)$$ equals
\begin{align*} & \exp(-2\pi^2\alpha^2T) \cdot \left(\sum_{t=0}^{p_2-1} \frac{\left(-2\pi^2\alpha^2\left(\sum_{i=1}^m S_i^2-T\right)\right)^t}{t!} \right)P(S_i)\\ & + O(1)^p\left(\left( \alpha^2\left(\sum_{i=1}^m S_i^2-T\right)\right)^{p/2}+\left( \alpha^2\left(\sum_{i=1}^m S_i^2-T\right)\right)^{p}+\left(\sum_{i=1}^m (\alpha S_i)^{4}\right)^{p/8}+\left(\sum_{i=1}^m (\alpha S_i)^{4}\right)^{p/2}\right).
\end{align*}
The claim now follows.
\end{proof}
\subsection{Final analysis}
We can finally state our main generator and prove Theorem \ref{th:mainintrohs}.
\begin{proof}[Proof of Theorem \ref{th:mainintrohs}]
Let $Y_1,Y_2$ be the generators from Propositions \ref{genLargeProp} and \ref{SmallAlphaProp} for $\delta = \epsilon/6n$. Let $Y$ be the the coordinate-wise product of the strings $Y_1,Y_2$. We claim that for any $v \in \zpm^n$ and $X \in_u \dpm^n$,
\begin{equation}
  \label{eq:mainproof1}
\dtv(v\cdot X,v\cdot Y) \leq \epsilon.
\end{equation}
The theorem follows immediately from the above claim and the bounds on the seed-lengths from Propositions \ref{genLargeProp} and \ref{SmallAlphaProp}.

To prove the theorem, we first prove that for all $\alpha \in \R$,
$$
\left|\E[\phi_{v,\alpha}(X)]-\E[\phi_{v,\alpha}(Y)] \right| \leq \epsilon/(2n).
$$
Now, if $\log^3(1/\delta)/\nmt{v} \leq \alpha$, then
$$
\E[\phi_{v,\alpha}(Y)] = \E[\phi_{D(Y_2)v,\alpha}(Y_1)]
$$
and
$$
\E[\phi_{v,\alpha}(X)] = \E[\phi_{v,\alpha}(D(Y_2)X)] = \E[\phi_{D(Y_2)v,\alpha}(X)].
$$
However by Proposition \ref{genLargeProp}, we have that
$$
\left|\E[\phi_{D(Y_2)v,\alpha}(Y_1)] - \E[\phi_{D(Y_2)v,\alpha}(X)] \right| \leq \epsilon/(3n).
$$
Similarly, if $\alpha \leq \log^3(1/\delta)/\nmt{v}$, then
then note that
$$
\E[\phi_{v,\alpha}(Y)] = \E[\phi_{D(Y_1)v,\alpha}(Y_2)]
$$
and
$$
\E[\phi_{v,\alpha}(X)] = \E[\phi_{v,\alpha}(D(Y_1)X)] = \E[\phi_{D(Y_1)v,\alpha}(X)].
$$
However by Proposition \ref{SmallAlphaProp}, we have that
$$
\left|\E[\phi_{D(Y_1)v,\alpha}(Y_2)] - \E[\phi_{D(Y_1)v,\alpha}(X)] \right| \leq \epsilon/(3n).
$$
Thus, we have our result for all $\alpha\in[0,1/4]$. Noting that $\phi_{v,-\alpha}(X) = \overline{\phi_{v,\alpha}(X)}$, we determine that the statement in question holds for $\alpha$ if and only if it holds for $-\alpha$. Thus, the inequality in question holds for all $\alpha\in[-1/4,1/4]$. Next, note that for any $X\in \dpm^n$ that $\phi_{v,\alpha+1/2}(X) = \exp(\pi i v\cdot X)\phi_{v,\alpha}(X) = (-1)^{\|v\|_0}\phi_{v,\alpha}(X)$. Thus, the statement in question holds for $\alpha$ if and only if it holds for $\alpha+1/2$. Thus, it holds for all real $\alpha$. Equation \ref{eq:mainproof1} now follows from the above argument and Claim \ref{clm:introfouriertotv} applied to $Z_1 = \dotp{v}{X}$ and $Z_2 = \dotp{v}{Y}$. 
\ignore{
Equation \ref{eq:mainproof1} now follows from noting that
\begin{align*}
\left|\pr(v\cdot X=k) - \pr(v\cdot Y=k) \right| & = \left|\int_{0}^1 e^{-2\pi i k \alpha}\left(\E[\phi_{v,\alpha}(X)]-\E[\phi_{v,\alpha}(Y)]\right)d\alpha \right|\\
& \leq\int_{0}^1 \left|\E[\phi_{v,\alpha}(X)]-\E[\phi_{v,\alpha}(Y)]\right|d\alpha\\
& \leq \epsilon/(3n).
\end{align*}
Thus
$$
\dtv(v\cdot X,v\cdot Y) = \sum_{k=-n}^n\left|\pr(v\cdot X=k) - \pr(v\cdot Y=k) \right| \leq \epsilon.
$$}
\end{proof}
\ignore{
\begin{thm}\label{mainThm}
Let $n,\epsilon>0$. Let $Y_1$ and $Y_2$ be the generators from Propositions \ref{genLargeProp} and \ref{SmallAlphaProp} for $\delta = \epsilon/(3n)$. Let $Y$ be given by the pointwise product $Y=Y_1Y_2$. Then for any $v\in\zpm^n$ and $X\in_u\dpm^n$, we have that
$$
\dtv(v\cdot X,v\cdot Y) \leq \epsilon.
$$
\end{thm}

Now, by Claim \ref{clm:introfouriertotv}, to prove the above claim it suffices to show that for all $\alpha \in [-1/4,1/4]$,
$$
\left|\E[\phi_{v,\alpha}(X)]-\E[\phi_{v,\alpha}(Y)] \right| \leq \epsilon/(2n).
$$

\begin{proof}
If $\log^3(1/\delta)/\|v\|_0 \geq \alpha \geq 0$, then note that
$$
\E[\phi_{v,\alpha}(Y)] = \E[\phi_{D(Y_1)v,\alpha}(Y_2)]
$$
and
$$
\E[\phi_{v,\alpha}(X)] = \E[\phi_{v,\alpha}(D(Y_1)X)] = \E[\phi_{D(Y_1)v,\alpha}(X)].
$$
However by Proposition \ref{SmallAlphaProp}, we have that
$$
\left|\E[\phi_{D(Y_1)v,\alpha}(Y_2)] - \E[\phi_{D(Y_1)v,\alpha}(X)] \right| \leq \epsilon/(3n).
$$

Similarly, if $1/4 \geq \alpha \geq \log^3(1/\delta)/\|v\|_0 $, then note that
$$
\E[\phi_{v,\alpha}(Y)] = \E[\phi_{D(Y_2)v,\alpha}(Y_1)]
$$
and
$$
\E[\phi_{v,\alpha}(X)] = \E[\phi_{v,\alpha}(D(Y_2)X)] = \E[\phi_{D(Y_2)v,\alpha}(X)].
$$
However by Proposition \ref{genLargeProp}, we have that
$$
\left|\E[\phi_{D(Y_2)v,\alpha}(Y_1)] - \E[\phi_{D(Y_2)v,\alpha}(X)] \right| \leq \epsilon/(3n).
$$
Thus, we have our result for all $\alpha\in[0,1/4]$. Noting that $\phi_{v,-\alpha}(X) = \overline{\phi_{v,\alpha}(X)}$, we determine that the statement in question holds for $\alpha$ if and only if it holds for $-\alpha$. Thus, the inequality in question holds for all $\alpha\in[-1/4,1/4]$. Next, note that for any $X\in \dpm^n$ that $\phi_{v,\alpha+1/2}(X) = \exp(\pi i v\cdot X)\phi_{v,\alpha}(X) = (-1)^{\|v\|_0}\phi_{v,\alpha}(X)$. Thus, the statement in question holds for $\alpha$ if and only if it holds for $\alpha+1/2$. Thus, it holds for all real $\alpha$.
\end{proof}

Theorem \ref{mainThm} now follows from noting that
\begin{align*}
\left|\pr(v\cdot X=k) - \pr(v\cdot Y=k) \right| & = \left|\int_{0}^1 e^{-2\pi i k \alpha}\left(\E[\phi_{v,\alpha}(X)]-\E[\phi_{v,\alpha}(Y)]\right)d\alpha \right|\\
& \leq\int_{0}^1 \left|\E[\phi_{v,\alpha}(X)]-\E[\phi_{v,\alpha}(Y)]\right|d\alpha\\
& \leq \epsilon/(3n).
\end{align*}
Thus
$$
\dtv(v\cdot X,v\cdot Y) = \sum_{k=-n}^n\left|\pr(v\cdot X=k) - \pr(v\cdot Y=k) \right| \leq \epsilon.
$$}


\bibliographystyle{alpha}
\bibliography{references}

\appendix
\section{Proofs from Section \ref{sec:prelims}}
\label{sec:appendix}

\begin{proof}[Proof of Lemma \ref{hcBiasedLem}]
This follows from the fact that $\|Q^p\|_1 \leq \|Q\|_1^p.$ Therefore,
\[
\E\sbkets{ Q(x)^p} \leq \E_{X\in_u \{\pm 1\}^n}[Q(X)^p] +\|Q\|_1^p
\delta\leq (p-1)^{p d/2} \cdot \|Q\|_2^p + \|Q\|_1^p \delta. 
\]
\end{proof}

\begin{proof}[Proof of Lemma \ref{lm:epschernoff}]
Note that because $\|v\|_2=1$ that $\|v\|_1\leq \|v\|_0^{1/2}$ by Cauchy-Schwarz. We note by the Markov inequality that for even $p$ that
$$
\pr[|\iprod{v}{x}| > t] \leq t^{-p}\E[|\iprod{v}{x}|^p].
$$
We need a slightly strengthened version of Lemma \ref{hcBiasedLem} to bound this.
Note that if $f(x)=\iprod{v}{x}$
$$
\E[f(x)] \leq \|f^p\|_0\epsilon + \|f\|_p^p \leq \|v\|_0^p \epsilon + (p-1)(p-3)\cdots 1.
$$
The bound on $\|f\|_p$ comes from noting that the expectation of $f^p$ under Gaussian inputs is $(p-1)(p-3)\cdots 1$ and that the expectation under Bernoulli inputs is at most this (which can be seen by expanding and comparing terms). Therefore, we have that
$$
\pr[|\iprod{v}{x}| > t] \leq t^{-p}\sqrt{2}(p/e)^{p/2} + t^{-p}\|v\|_0^p \epsilon
$$
Letting $p$ be the largest even integer less than $t^2$, we find that this is at most
$$
\sqrt{2} \exp(-p/2)+ \|v\|_0^{t^2} \epsilon,
$$
which is sufficient when $t\geq 2$. For $1\leq t \leq \sqrt{2}$, the
trivial upper bound of $1$ is sufficient, and for $\sqrt{2}\leq t \leq
2$, we may instead use the bound for $p=2$. 
\end{proof}

\begin{proof}[Proof of Lemma \ref{hashmomentsLem}]
Let $I_{i,k}$ be the indicator function of the event that
$h(i)=k$. Note that $h(v)=\sum_{i,j,k} I_{i,k}I_{j,k}v_i^2 v_j^2.$
Therefore,
$$
h(v)^p = \sum_{i_1,\ldots,i_{p},j_1,\ldots,j_p}\sum_{k_1,\ldots,k_p} \prod_{t=1}^p I_{i_{t},k_t}I_{j_{t},k_t} \prod_{t=1}^{p} v_{i_t}^2v_{j_t}^2.
$$
Let $R(i_t,j_t,k_t)$ be $0$ if for some $t,t'$ $k_t\neq k_t'$ but one
of $i_{t}$ or $j_t$ equals $i_{t'}$ or $j_{t'}$ and otherwise be equal
to $m^{-T}$ where $T$ is the number of distinct values taken by $i_t$
or $j_t$. Notice that by the $\delta$-biasedness of $h$ that
$$
\E\left[\prod_{t=1}^p I_{i_{t},k_t}I_{j_{t},k_t}\right] \leq R(i_t,j_t,k_t) + \delta.
$$
Combining with the above we find that
\begin{align*}
\E[h(v)^p] & \leq \sum_{i_1,\ldots,i_{p},j_1,\ldots,j_p}\sum_{k_1,\ldots,k_p} (R(i_t,j_t,k_t) + \delta)\prod_{t=1}^{p} v_{i_t}^2v_{j_t}^2\\
& \leq \sum_{i_1,\ldots,i_{p},j_1,\ldots,j_p}\sum_{k_1,\ldots,k_p} R(i_t,j_t,k_t)\prod_{t=1}^{p} v_{i_t}^2v_{j_t}^2 + \delta m^p\sum_{i_1,\ldots,i_{p},j_1,\ldots,j_p}\prod_{t=1}^{p} v_{i_t}^2v_{j_t}^2\\
& \leq \sum_{i_1,\ldots,i_{p},j_1,\ldots,j_p}\sum_{k_1,\ldots,k_p} R(i_t,j_t,k_t)\prod_{t=1}^{p} v_{i_t}^2v_{j_t}^2 + \delta m^p \|v\|_2^{4p}.
\end{align*}
Next we consider
$$
\sum_{k_1,\ldots,k_p} R(i_t,j_t,k_t)
$$
for fixed values of $i_1,\ldots,i_p,j_1,\ldots,j_p$. We claim that it
is at most $m^{-S/2}$ where $S$ is again the number of distinct
elements of the form $i_t$ or $j_t$ that appear in this way an odd
number of times. Letting $T$ be the number of distinct elements of the
form $i_t$ or $j_t$, the expression in question is $m^{-T}$ times the
number of choices of $k_t$ so that each value of $i_t$ or $j_t$
appears with only one value of $k_t$. In other words this is $m^{-T}$
times the number of functions $f:\{i_t,j_t\}\rightarrow[m]$ so that
$f(i_t)=f(j_t)$ for all $t$. This last relation splits $\{i_t,j_t\}$
into equivalence classes given by the transitive closure of the
operation that $x\sim y$ if $x=i_t$ and $y=j_t$ for some $t$. We note
that any $x$ that appears an odd number of times as an $i_t$ or $j_t$
must be in an equivalence class of size at least $2$ because it must
appear at least once with some other element. Therefore, the number of
equivalence classes, $E$ is at least $T-S/2$. Thus, the sum in
question is at most $m^{-T}m^E \leq m^{-S/2}$. Therefore, we have that 
$$
\E[h(v)^p] \leq (2p)!\sum_{\textrm{Multisets }M\subset [n],
  |M|=2p}m^{-\{\mathrm{Odd}(M)\}/2}\prod_{i\in M} v_i^2 +\delta m^p
\|v\|_2^{4p}. 
$$
Where $\mathrm{Odd}(M)$ is the number of elements occurring in $M$ an
odd number of times. This equals 
\begin{align*}
\E[h(v)^p] & \leq (2p)!\sum_{k=0}^p\sum_{\textrm{Multisets }M\subset
  [n], |M|=2p, \mathrm{Odd}(M)=2k}m^{-k}\prod_{i\in M} v_i^2 +\delta
m^p \|v\|_2^{4p}\\ 
& \leq (2p)!\sum_{k=0}^p m^{-k}
\sum_{i_1,\ldots,i_{2k}}\sum_{j_1,\ldots,j_{p-k}}\prod v_{i_t}^2 \prod
v_{j_t}^4+\delta m^p \|v\|_2^{4p}\\ 
& = (2p)!\sum_{k=0}^p \left(\frac{\|v\|_2^4}{m}\right)^k
\|v\|_4^{4(p-k)}+\delta m^p \|v\|_2^{4p}\\ 
& \leq O(p)^{2p} \left(\frac{\|v\|_2^4}{m}\right)^p + O(p)^{2p}
\|v\|_4^{4p}+\delta m^p \|v\|_2^{4p}. 
\end{align*}
Note that the second line above comes from taking $M$ to be the
multiset 
$$\{i_1,i_2,\ldots,i_{2k},j_1,j_1,j_2,j_2,\ldots,j_{p-k},j_{p-k}\}.$$ 
This completes our proof.
\end{proof}

\begin{proof}[Proof of Lemma \ref{lm:hashing}]
 Let $X_i$ denote the indicator random variable which is $1$ if $h(i) = j$ and $0$ otherwise. Let $Z = \sum_i v_i X_i$. Now, if $h$ were a truly random hash function, then, by Hoeffding's inequality,
$$\pr\sbkets{| Z - \nmo{v}/m| \geq t} \leq 2 \exp\bkets{- t^2/2\sum_i v_i^2}.$$
Therefore, for a truly random hash function and even integer $p \geq 2$, $\nmp{Z} = O(\nmt{v}) \sqrt{p}$. Therefore, for a $\delta$-biased hash family, we get $\nmp{Z}^p \leq O( p)^{p/2} \nmt{v}^p + \nmo{v}^p \delta$. Hence, by Markov's inequality, for any $t > 0$,
$$\pr\sbkets{ |Z - \nmo{v}/m| \geq t} \leq \frac{O(p)^{p/2} \nmt{v}^p + \nmo{v}^p \delta}{t^p}.$$
\end{proof}

\end{document}